%% file: ConsensusArxivVer.tex
\documentclass[format=acmsmall, review=false]{acmart}
\usepackage{acm-ec-20}
\usepackage[latin1]{inputenc}
\usepackage{amsmath}
\usepackage{amsfonts}
\usepackage{amssymb}
\usepackage{graphicx}
\usepackage{enumitem}
\usepackage{amsthm}
\usepackage{color}
\usepackage[english]{babel}
\usepackage{setspace}

\theoremstyle{definition}
\newtheorem{theorem}{Theorem}[section]
\newtheorem{lemma}[theorem]{Lemma}
\newtheorem{claim}[theorem]{Claim}

\newtheorem{corollary}[theorem]{Corollary}

\newtheorem*{theorem*}{Theorem}

\newcommand{\cl}{{\bf cl}}

\newcommand{\E}{\mathbb{E}}

\newcommand{\conv}{\text{conv}}
\newcommand{\vol}{\text{vol}}

\newcommand{\overbar}[1]{\mkern 1.5mu\overline{\mkern-1.5mu#1\mkern-1.5mu}\mkern 1.5mu}

\title{How Many Freemasons Are There?\\ The Consensus Voting Mechanism in Metric Spaces}

\author{Mashbat Suzuki}
\affiliation{\institution{McGill University}
	}
\email{mashbat.suzuki@mail.mcgill.ca}
\author{Adrian Vetta}
\affiliation{\institution{McGill University}
	}
\email{adrian.vetta@mcgill.ca}

\begin{abstract}
	We study the evolution of a social group when admission to the group is determined via {\em consensus} or {\em unanimity voting}.
	In each time period, two candidates apply for membership and a candidate is selected if and only if all the current group members agree.
	We apply the {\em spatial theory of voting} where group members and candidates are located in a metric space and
	each member votes for its closest (most similar) candidate.
	Our interest focuses on the expected cardinality of the group after $T$ time periods. 
	To evaluate this we study the geometry inherent in dynamic consensus voting over a metric space.
	This allows us to develop a set of techniques for lower bounding and upper bounding the expected cardinality of a group. 
	We specialize these methods for two-dimensional metric spaces. For the unit ball the expected cardinality 
	of the group after $T$ time periods is $\Theta(T^\frac{1}{8})$.
	In sharp contrast, for the unit square the expected cardinality is at least $\Omega(\ln T)$ but at most $O(\ln T \cdot \ln\ln T )$.
\end{abstract}

\begin{document}
	
\maketitle
\input{introduction.tex}

\input{geometry.tex}

\input{tools.tex}

\input{lower.tex}

\input{upper.tex}

\input{conclusion.tex}

\bibliographystyle{plain}
\bibliography{FL}{}	
\end{document}

%% file: introduction.tex
	\section{Introduction}\label{sec:intro}
	This paper studies the evolution of social groups over time. In an {\em exclusive social group}, the
	existing group members vote to determine whether or not to admit a new member.
	Familiar examples include the freemasons, fraternities, membership-run sports and social clubs,
	acceptance to a condominium, as well as academia.
	To analyze the inherent dynamics we use the model of Alon, Feldman, Mansour, Oren and Tennenholtz~\cite{dyn}.
	In each time period, two candidates apply for membership and the current members vote to decide if either
	or none of them is acceptable. The spatial model of voting is used: each person is located uniformly at random in a metric space 
	and each group member votes for the candidate closest to them.
	
	Alon et al.~\cite{dyn} analyze social group dynamics in a one-dimensional metric space, specifically, the unit interval $[0,1]$.
	They examine how outcomes vary under different winner determination rules, in particular,   
	majority voting and consensus voting.
	In {\em consensus voting} or {\em unanimity voting} a candidate is elected if and only if the group members agree unanimously.
	Equivalently, every member may {\em veto} a potential candidate. 
	
	Our interest lies in the evolution of the group size under consensus voting; that is, what is the expected cardinality of the social group $G^T$ after $T$ time periods? In the one-dimensional setting the answer is quite simple. There, Alon et al.~\cite{dyn}
	show that under consensus voting {\em if} a candidate is elected in round $t$ then, with high probability, it is within a distance $\Theta(1/\sqrt{t})$ 
	of one endpoint of the interval. Because the winning candidate must be closer to all group members than the losing candidate, both candidates
	must therefore be near the endpoints. This occurs with probability $\Theta(1/t)$.
	As a consequence, in a one-dimensional metric space, the expected size of the social group after $T$ time periods is
	$\E[ |G^{T}| ] \cong  \ln T$. Here we use the notation $f \cong g$ if both $f \lesssim g$ and $g \lesssim f$, where
	$f \lesssim g$ if $f\le c\cdot g$ for some constant $c$. 
	
	Bounding the expected group size in higher-dimensional metric spaces is more complex and is the focus of this paper.
	To do this, we begin in Section~\ref{sec: geometry} by examining the geometric aspects of consensus voting in higher-dimensional 
	metric spaces. More concretely, we explain how winner determination relates to the convex hull of the group members
	and the Voronoi cells formed by the candidates.
	This geometric understanding enables us to construct, in Section~\ref{sec: MainTools}, a set of techniques, based upon {\em cap methods} 
	in probability theory, that allow for the upper bounding and lower bounding of expected group size. 
	In Sections~\ref{sec: LB} and~\ref{sec: UB}, we specialize these techniques to two-dimensions for application 
	on the fundamental special cases of the unit square and the unit ball.
	Specifically, for the unit square we show the following lower and upper bounds on expected group size. 
	\begin{theorem}
		Let the metric space be the unit square $\mathbb{H}$. Then the expected cardinality of the social group 
		after $T$ periods is bounded by \ $\ln T\ \lesssim \ \E[\, |G^{T}|\, ]  \ \lesssim \ \ln T \cdot \ln\ln T$.
	\end{theorem}
	Thus, expected group size for the two-dimensional unit square is comparable to that of the one-dimensional interval.
	Surprisingly, there is a dramatic difference in expected group size between the unit square and the unit ball.
	For the unit ball, the expected group size evolves not logarithmically but polynomially with time.
	\begin{theorem}
		Let the metric space be the unit ball $\mathbb{B}$. Then the expected cardinality of the social group 
		after $T$ periods is \ $\E[\, |G^{T}|\, ] \ \cong\ T^{\frac{1}{8}}$.
	\end{theorem}

	\subsection{Background and Related Work}
	Here we discuss some background on the spatial model and consensus voting.
	The spatial model of voting utilized in this paper dates back nearly a century to the celebrated work of Hotelling~\cite{harold1929stability}. 
	His objective was to study the division of a market in a duopoly when consumers are distributed over a one-dimension space,
	but he noted his work had intriguing implications for electoral systems.
	Specifically, in a two-party system there is an incentive for the political platforms of the two parties to {\em converge}.
	This was formalized in the {\em median voter theorem} of Black~\cite{Bla48}: in a one-dimensional ideological space
	the location of the median voter induces a Condorcet winner\footnote{A candidate is a {\em Condorcet winner} if, in a pairwise majority vote, it 
		beats every other candidate.}, given single-peaked voting preferences. The traditional voting assumption in a metric space
	is {\em proximity voting} where each voter supports its closest candidate; observe that 
	proximity voting gives single-peaked preferences.
	
	
	The {\em spatial model of voting} was formally developed by Downs~\cite{Dow57} in 1957, again in a one-dimensional metric space.
	Davis, Hinich, and Ordeshook~\cite{DHO70} expounded on practical necessity of moving beyond just one dimension.
	Interestingly, they observed that in two-dimensional metric spaces, a Condorcet winner is not guaranteed even with proximity voting.
	Of particular relevance here is their finding that, in dynamic elections, the order in which candidates are considered can 
	fundamentally affect the final outcome~\cite{Bla48,DHO70}.
	
	There is now a vast literation on spatial voting, especially concerning the strategic aspects of simple majority voting; see, for example, 
	the books~\cite{enelow1984spatial,arrow1990advances,merrill1999unified, Poo05, Sch07}. There has also been a vigorous debate concerning
	whether voter utility functions in spatial models should be distance-based (such as the standard assumption of proximity voting used here), 
	relational (e.g. directional voting~\cite{RS89}), or combinations thereof~\cite{merrill1999unified}.
	This debate has been philosophical, theoretical and experimental~\cite{Gro73, Mat79, LK00, TV08, Cla09, LP10}.
	Recently there has also been a large amount of interest in the spatial model by the artificial intelligence 
	community~\cite{ABJ15, FFG16, AJ17, SE17, BLS19}.
	It is interesting to juxtapose these modern potential applications with the original motivations 
	suggested by Black~\cite{Bla48}, such as the administration of colonies!
	
	Consensus is one of the oldest group decision-making procedures. In addition to exclusive social groups, it is familiar
	in a range of disparate settings, including judicial verdicts, Japanese corporate governance \cite{vogel1975modern}, and 
	even decision making in religious groups, such as the Quakers \cite{hare1973group}. From a theoretic perspective, 
	consensus voting in a metric spaces has also been studied by Colomer~\cite{colomer1999geometry} who highlights 
	the importance the initial set of voters can have on outcomes in a dynamic setting.
	

%% file: geometry.tex
	\section{The Geometry of Consensus Voting}\label{sec: geometry}
	In this section, we present a simple geometric interpretation of a single election using consensus voting in the spatial model.
	In the subsequent sections, we will apply this understanding, developed for the static
	case, to study the dynamic model. Specifically, we examine how a group grows over time
	when admission to the group is via a sequence of consensus elections.
	
	Let $G^0=\{v_1,\cdots, v_k\}$ denote the initial set of group members\footnote{We may take the cardinality of the initial group to be any 
	constant $k$. In particular, we may assume $k=1$.}, selected uniformly and independently from
	a metric space $K$. 
	In the consensus voting mechanism, for each round $t\ge 1$, a finite set of candidates $C^t=\{w_1,\cdots, w_n\}\subseteq K$ applies 
	for membership. Members of the group at the start of round $t$, denoted $G^{t-1}$, are eligible to vote. 
	Assuming the {\em spatial theory of voting}, each group member will vote for the candidate who is closest to her in the metric space. 
	That is, member $v_i$ votes for candidate $w_j$ if and only if $d(v_i,w_j)\leq d(v_i,w_k)$ for every candidate $w_k\neq w_j$. 
	Under the {\em consensus (unanimity) voting rule}, if {\em every} group member selects the candidate $w_j\in C^t$ then $w_j$ is 
	accepted to the group and $G^t=G^{t-1}\cup w_j$; otherwise, if the group does not vote unanimously then
	no candidate wins selection and $G^t=G^{t-1}$. 
		
	As stated, to study how group size evolves over time, our first task is to develop a more precise understanding 
	of when a candidate will be selected under consensus voting in a single election. Fortunately, there is a nice 
	geometric characterization for this property in terms of the \textit{Voronoi cells} (regions) formed in the metric 
	space $K$ by the candidates (points) $C=\{w_1,\cdots w_n\}$.
	Specifically, the Voronoi cell $H_i$ associated with point $w_i$ is 
	$H_i:=\{v\in K \ | \ d(v,w_i)\leq d(v,w_j) \ \text{for all } i\neq j\}$.
	We will see that the convex hull of the group members $G\subseteq K$, which we denote $S=conv(G)$, 
	plays an important role in winner determination. 
	The characterization theorem for the property that a candidate is selected under the consensus voting mechanism is then:
	\begin{theorem}\label{thm:voronoi}
		Let $C=\{ w_1,w_2,\dots, w_n\}$ be the candidates and let $H_1,H_2,\dots,H_n$ be the Voronoi cells on $K$ 
		generated by $C$. Then there is a winner under consensus voting {\em if and only if}
		$S=conv(G)\subseteq H_i$ for some candidate $w_i$.
	\end{theorem}
	
	\begin{proof} Assume $S=\subseteq H_i$ for some candidate $w_i\in C$. Then, for every voter $v_j \in G$,  we have 
		$d(v_j,w_i)\leq d(v_j,w_k)$ for any other candidate $w_k\in C$. 
		Hence, every voter prefers candidate $w_i$ over all the other candidates. Thus candidate $w_i$ is selected.
		Conversely, assume that candidate $w_i$ is selected. Then, by definition of consensus voting, each voter $v_j \in G$ 
		voted for $w_i$. Thus $d(v_j,w_i)\leq d(v_j,w_k)$ for all $k\neq i$. Ergo, $G\subseteq S\subseteq H_i$.	
	\end{proof}
	
	Several useful facts can be derived from this characterization. These facts are stated in the subsequent corollary and lemma. 
	\begin{corollary}\label{cor:subset} Let $C=\{ w_1,w_2,\dots, w_n\}$ be set of candidates. If there is a candidate 
		accepted with $S=A$, then the same candidate is also accepted  with $S=B$ for any convex set $B\subseteq A$. 
	\end{corollary}
	\begin{proof}
	By Theorem~\ref{thm:voronoi}, a candidate is accepted if and only if the current convex hull is 
	entirely contained within one of Voronoi regions, $H_1, H_2,\dots, H_n$, generated by the candidates.
	Clearly, if $A \subseteq H_i$ then  $B \subseteq H_i$. Therefore, if there is an 
	acceptance with $S=A$ then there would be an acceptance with $S=B$. The result follows immediately.
\end{proof}
	The next lemma requires the following definition: let $B(v, w)$ denote the 
	{\em Euclidean ball} centred at $v$ with radius $\|w-v\|$. Furthermore, we denote by~$\partial S$ 
	the set of vertices (extreme points) of the convex hull~$S$ of the group members. Observe that $\conv(\partial S) = S$ and that $\partial S \subseteq G$.

	\begin{lemma}\label{lem: BallIntersect}	
		Let $G$ be current set of group members and $S$ be its convex hull. Let $C=\{w_1, w_2, \cdots, w_n  \}$ be the current candidates. 
		Under consensus, there is a winning candidate if and only if  $\exists w_i\in C$ such that 
		$$w_i\in \bigcap\limits_{k\in  [n]\setminus i}\, \bigcap\limits_{v_j\in \partial S} B(v_j, w_k)$$
	\end{lemma}
	\begin{proof}[Proof of Lemma~\ref{lem: BallIntersect}]
		If there is a consensus then there is a $w_i\in C$ who is selected. Hence among all candidates, $w_i$ is closest to each group member.  
		That is, $d(v_j,w_i)\leq d(v_j,w_k)$, for each group member $v_j\in G$ and each candidate $w_k\in C\setminus w_i$. It follows that 
		$$w_i\in \bigcap\limits_{k\in  [n]\setminus i}\, \bigcap\limits_{v_j\in G} B(v_j, w_k)  
		\subseteq \bigcap\limits_{k\in  [n]\setminus i} \, \bigcap\limits_{v_j\in \partial S} B(v_j, w_k) $$ 
		where the last inclusion holds since $\partial S\subseteq G$.
		Conversely, suppose there exists a $w_i$ such that $w_i\in\bigcap\limits_{k\in  [n]\setminus i}\bigcap\limits_{v_j\in \partial S} B(v_j, w_k) $. 
		Thus, $w_i$ satisfies    $d(v_j,w_i)\leq d(v_j,w_k)$ for each voter $v_j\in \partial S$ and $w_k\in C \setminus w_i$. Hence we 
		have $\partial S\subseteq H_i$, which implies $\conv(\partial S)=S \subseteq \conv(H_i)=H_i $ as $H_i$ is convex. 
		Therefore $G\subseteq S\subseteq H_i$. Thus, by Theorem \ref{thm:voronoi}, candidate 
		$w_i$ wins under consensus voting. 
	\end{proof}

	Following Alon et al.~\cite{dyn}, from now on we restrict attention to case of $n=2$ candidates in each round.
The case $n\ge 3$ is not conceptually harder and the ideas presented in this paper do extend to that setting, but mathematically 
the analyses would be even more involved than those that follow.

%% file: tools.tex
	\section{General Tools for Bounding Expected Group Size}\label{sec: MainTools}
	
	In this section we introduce a general approach for obtaining both upper and lower bounds 
	on the expected cardinality of the social group in round $t$. These techniques apply for consensus voting in any convex compact domain $K$. In the rest of the paper we will specialize these methods 
	for the cases in which $K$ is either a unit ball or a unit square. In particular, lower bounds are provided 
	for these two domains in Section~\ref{sec: LB} and upper bounds in Section~\ref{sec: UB}.
	
	Let $K$ be a convex compact set,  and let $C^t=\{w_1,w_2\}$ be candidates appearing in round $t$. 
	We assume each candidate $w_i$ is distributed uniformly on $K$. We may also assume that $\vol(K)=1$, 
	as otherwise we can absorb the associated constant factor into our bounds. Note that the expected group size is $\E[|G^T|]=\sum_{t=1}^T \Pr[X^t]$, where $X^t$ denotes the event a new candidate wins in round $t$.   
	
	Let's first present the intuition behind our approach to upper bounding the probability of selecting a 
	candidate in any round. Recall that, by Theorem~\ref{thm:voronoi}, given two candidates $\{w_1,w_2\}$ in round $t+1$, 
	we accept candidate $i$ if and only if $S^t\subseteq H_i(w_1,w_2)$. Now in order for the convex hull to satisfy 
	$S^t\subseteq H_i(w_1,w_2)$, it must be the case that in the previous round (i) $S^{t-1}\subseteq H_i(w_1,w_2) $, and (ii) a 
	new candidate did not get accepted inside the complement $\overbar{H_i(w_1,w_2)}= \cl(K\setminus H_i(w_1,w_2))$, where $\cl$~denotes set closure. 
	Applying this argument recursively with respect to the worst case convex hulls for accepting candidates inside $\overbar{H_i(w_1,w_2)}$, we will 
	obtain an upper bound on the probability of accepting a candidate. Such worst case convex hulls can be found by appropriately applying Corollary~\ref{cor:subset}. 
	
	To formalize this intuition, we require some more notation. We denote by $Z(H)$ the event that a new candidate is selected inside $H$. Let $ \Pr[Z(H)|S=A ]$ denote the probability of selecting a candidate inside $H$ given the convex hull of the group members is $A$.
	The shorthand $ \Pr[Z(H)|A]=\Pr[Z(H)|S=A]$ will be used when the context is clear. Note, by Lemma~\ref{lem: BallIntersect}, 
	the probability of acceptance depends only on the shape of the convex hull of the members, and not on the round. 
	That is, if $S^t=S^{\hat{t}}=A$ for two rounds $t\neq \hat{t}$ then the probabilities of accepting a candidate inside a 
	given region in the rounds $t+1$ and $\hat{t}+1$ are exactly the same. 
	
	We say a set $A$ is a {\em cap} if there exists a half space $W$ such that  $A=K\cap W$. We remark that caps have been 
	widely used for studying the convex hull of random points; see the survey article~\cite{baddeley2007random} and the references therein. 
	Of particular relevance here is that, in the case of two candidates, the Voronoi regions for the candidates are caps.  
	Furthermore, $\overbar{H_1(w_1,w_2)}=H_2(w_1,w_2)$ and vice versa. 
	
	\begin{theorem}\label{thm: GeneralBnd}
		Let $K$ be convex compact domain and let $f_K(w_1,w_2)$ be any function which 
		satisfies $f_K(w_1,w_2)\leq \min\limits_{i}\Pr[Z(H_i(w_1,w_2))\ | \ \overbar{H_i(w_1,w_2)}]$. Then 
		$$\Pr[X^{t+1}]\ \lesssim \ \int_K \int_K e^{-tf_K(w_1,w_2)} \, dw_1\, dw_2$$
	\end{theorem}	 
	\begin{proof}
		Observe that, for any cap $A$, we have the following inequality:
		\begin{align}
		\Pr[S^t \subseteq  A ] \ &=\ \left(1-\Pr[Z(\overbar{A})\ | \ S^{t-1} \subseteq A] \right) \cdot \Pr[S^{t-1}\subseteq A] \nonumber \\ 
		&\ \leq\ \left(1-\Pr[Z(\overbar{A})\ | \ S^{t-1}=A  ] \right) \cdot \Pr[S^{t-1}\subseteq A] \nonumber\\
		&\ \leq\ \left(1-\Pr[Z(\overbar{A})\ | \ A \   ] \right)^t  \label{eq: bndCap}
		\end{align}
		Here the first inequality follows from Corollary~\ref{cor:subset}. The second inequality is 
		obtained by repeating the argument inductively for $S^{t-1}$. Hence:
		\begin{align*}
		\Pr[X^{t+1}]
		&\ = \  \int\int \Pr[X^{t+1}  \ | \ (w_1,w_2) \ \text{are candidates}]  \, dw_1\,  dw_2  \\ 
		&\ =\  \int \int \big( \Pr[S^t\subseteq H_1(w_1,w_2)]+\Pr[S^t\subseteq H_2(w_1,w_2)] \big)  \,dw_1 \,dw_2 \\  
		&\ \leq\ \sum_{i=1}^2\, \int \int \big( 1-\Pr[Z(H_i(w_1,w_2)) \ | \ \overbar{H_i(w_1,w_2)} \ ]\big)^t\,dw_1 \,dw_2 \\ 
		&\ \leq\ 2\int \int (1-f_K(w_1,w_2))^t  \,dw_1 \,dw_2 \\ 
		&\ \leq\ 2\int \int  e^{-tf_K(w_1,w_2)}  \,dw_1 \,dw_2
		\end{align*}	
		The second equality holds by Theorem~\ref{thm:voronoi}. The first inequality follows by 
		combining inequality~(\ref{eq: bndCap}) and the facts $\overbar{H_1(w_1,w_2)}=H_2(w_1,w_2)$ and 
		$\overbar{H_2(w_1,w_2)}=H_1(w_1,w_2)$. Next, by assumption, 
		$f_K(w_1,w_2)\leq \Pr[Z(H_i(w_1,w_2))\ | \ \overbar{H_i(w_1,w_2)}]$ for each $i$; the last two inequalities follow immediately.	
	\end{proof}	
	As stated, Theorem~\ref{thm: GeneralBnd} allows us to upper bound the expected group size. 
	However, the theorem is not easily applicable on its own. To rectify this, consider the following easier to apply corollary.
	\begin{corollary}\label{cor: MainBound}  
		Let $K$ be compact space. If $f_K$ satisfies the conditions of 
		Theorem~\ref{thm: GeneralBnd} then
		$$\Pr[X^{t+1}]  \ \lesssim\  \frac{1}{t}+\int_0^{\frac{\ln(t)}{t}} te^{-t\lambda} \cdot \Phi(\lambda)\, d\lambda $$	
		where $\Phi(\lambda)= \int \int \mathbb{I}\left[f_K(w_1,w_2)\leq \lambda\right] dw_1 dw_2$.
	\end{corollary}
	\begin{proof}[Proof of Corollary~\ref{cor: MainBound}]  
		From Theorem~\ref{thm: GeneralBnd}, we have that,
		\begin{align}
		\Pr[X^{t+1}] &\ \lesssim\  \int \int \mathbb{I}\left[f_K(w_1,w_2)
		\leq \frac{\ln(t)}{t}\right] \cdot e^{-tf_K(w_1,w_2)} \, dw_1\, dw_2 \nonumber\\ 
		&\quad\quad\quad +\ \int \int \mathbb{I}\left[f_K(w_1,w_2)\geq \frac{\ln(t)}{t}\right] \cdot e^{-tf_K(w_1,w_2)} \, dw_1\, dw_2  \nonumber\\
		&\ \leq\  \int \int \mathbb{I}\left[f_K(w_1,w_2)\leq \frac{\ln(t)}{t}\right] \cdot e^{-tf_K(w_1,w_2)} \, dw_1 \,dw_2 \ +\  \frac{1}{t}~\label{ieq: tiq}
		\end{align}
		Next, consider the random variable $Y=f_K(w_1,w_2)$ and denote its {\em cumulative distribution function} by
		$F_Y(\lambda)=\Pr[Y\leq \lambda]=\Phi(\lambda)$. Then 
		\begin{align}
		\int \int \mathbb{I}\left[f_K(w_1,w_2)\leq \frac{\ln(t)}{t}\right] \cdot e^{-tf_K(w_1,w_2)} \, dw_1 \,dw_2   
			&\ =\  \E \left[\mathbb{I}\left[f_K(w_1,w_2)\leq \frac{\ln(t)}{t}\right] \cdot e^{-tf_K(w_1,w_2)}\right] \nonumber \\ 
			&\ =\ \E_Y\left[ \E \left[\mathbb{I}\left[Y\leq \frac{\ln(t)}{t}\right] \cdot e^{-tY}\ | \ Y \right]  \right] \nonumber  \\ 
			&\ =\ \int_0^{1}  \mathbb{I}\left[\lambda\leq \frac{\ln(t)}{t}\right] \cdot e^{-t\lambda}  \ dF_Y(\lambda)  \label{ieq: biq} 
		\end{align}
		The second equality is due to the {\em law of total expectation}. Now, because $F_Y$ is absolutely continuous it has a density function.  
		Combining~(\ref{ieq: tiq}) and (\ref{ieq: biq}), we then have that 
		\begin{align*}
		\Pr[X^{t+1}] \ &\lesssim\   \frac{1}{t}+ \int_0^{\frac{\ln(t)}{t}} e^{-t\lambda}\cdot  \frac{d}{d\lambda}\, \Phi(\lambda)\, d\lambda  \\ 
		&\ =\ \frac{1}{t}+\left[e^{-{t\lambda}}\Phi(\lambda)\right]_0^{\frac{\ln(t)}{t}}+\int_0^{\frac{\ln(t)}{t}} te^{-t\lambda} \cdot \Phi(\lambda)\, d\lambda \\ 
		&\ \lesssim\  \frac{1}{t} + \int_0^{\frac{\ln(t)}{t}} te^{-t\lambda} \cdot \Phi(\lambda)\, d\lambda
		\end{align*}
		Here the last inequality was obtained by noting $\Phi(0)=0$ and $\Phi(\lambda)\leq 1$.
	\end{proof}
	
	As alluded to earlier, when obtaining upper bounds for the unit ball and the unit square, we will apply Corollary~\ref{cor: MainBound}.
	Of course, in order to do this, we must find an appropriate function $f_K(w_1,w_2)$ which lower bounds the probability of 
	acceptance inside a Voronoi region $H_i(w_1,w_2)$, given the current convex hull is $\overbar{H_i(w_1,w_2)}$. 
	Finding such a function $f_K$ can require some ingenuity, but Lemma~\ref{lem: BndCap} below will be useful in assisting in 
	this task. Moreover, as we will see in Section~\ref{sec: LB}, 
	this lemma can be used to obtain lower bounds as well as upper bounds on the expected cardinality of the group.

	\begin{lemma}\label{lem: BndCap}
		Given a two-dimensional convex compact domain~$K$ and a cap~$A$. If $z_1$ and $z_2$ are the 
		endpoints of the line segment separating $A$ and $\overbar{A}$ then 	
		$$\Pr[Z(A) \ | \ \overbar{A} \ ] \ =\  2\cdot \int_A  \vol(B(z_1,\xi)\cap B(z_2,\xi) \cap A)    \ d\xi $$
	\end{lemma}
	\begin{proof}	
		Without loss of generality, let $w_2\in A$ be the selected candidate.  Thus, by Theorem~\ref{thm:voronoi}, 
		we have $\overbar{A}\subseteq H_2(w_1,w_2)$ because, by assumption, the convex hull of the voters is $S=\overbar{A}$. 
		We claim $w_1\notin \overbar{A}$. Otherwise, $w_1\in \overbar{A}\subseteq  H_2(w_1,w_2)$. But, by definition of the Voronoi regions, 
		we have $w_1\in H_1(w_1,w_2)$. Thus, $w_1\in H_1\cap H_2$.  
		However, this cannot happen unless $w_1=w_2$, a zero probability event. Therefore, we may assume that $w_1,w_2\in A $.
		
		Next, consider the line $l(z_1,z_2)$ separating $A$ and $\overbar{A}$; that is, $l(z_1,z_2)$ is the convex hull 
		$\conv(z_1,z_2)$. Given $w_1,w_2\in A $ we claim $l(z_1,z_2)\subseteq H_2(w_1,w_2)$ if and only if 
		$\overbar{A}\subseteq H_2(w_1,w_2)$. First assume $\overbar{A}\subseteq H_2(w_1,w_2)$. As $\bar{A}$ is 
		closed, $z_1$ and $z_2$ are in $\bar{A}$. Thus, by convexity, 
		$l(z_1,z_2)\subseteq \bar{A} \subseteq H_2(w_1,w_2)$.
		On the other hand, assume $l(z_1,z_2)\subseteq H_2(w_1,w_2)$ and consider the
		Voronoi edge defined by $E(w_1,w_2)=H_1(w_1,w_2)\cap H_2(w_1,w_2)$. Then 
		either $E(w_1,w_2)\subseteq A$ or $E(w_1,w_2)\subseteq \overbar{A}$ because $l(z_1,z_2)$ separates $A$ and $\overbar{A}$ 
		and, by assumption, $l(z_1,z_2)\subseteq H_2(w_1,w_2)$. Furthermore, observe  that the midpoint $\frac12(w_1+w_2)$ is 
		in $E(w_1,w_2)$ and in $A$. So it must be the case that $E(w_1,w_2)\subseteq A $. But if $E(w_1,w_2)\subseteq A $ then one of the 
		Voronoi regions contains $\overbar{A}$. In particular, $\overbar{A}\subseteq H_2(w_1,w_2)$ 
		because $l(z_1,z_2)\subseteq H_2(w_1,w_2)$. Therefore,
		\begin{align*}
		\Pr[Z(A) \ | \ \overbar{A} \ ] &\ = \ 2\int_A\int_A \mathbb{I}[\overbar{A}\subseteq H_2(w_1,w_2)]\ dw_2 \, dw_1  \\
		&\ = \ 2  \ \int_A\int_A \mathbb{I}[l(z_1,z_2)\subseteq H_2(w_1,w_2)] \ dw_2 \, dw_1 \\ 
		&\ =\ 2 \ \int_A \vol(B(z_1,w_1)\cap B(z_2,w_1) \cap A) \, dw_1
		\end{align*}
		Here the first equality holds since two candidates are equally likely to win by symmetry and, by 
		Theorem~\ref{thm:voronoi}, the winning candidate $w_2$ satisfies $\overbar{A}\subseteq H_2(w_1,w_2)$. 
		The second equality follows from the fact that $l(z_1,z_2)\subseteq H_2(w_1,w_2)$ if and only if 
		$\overbar{A}\subseteq H_2(w_1,w_2)$ when $w_1,w_2\in A$.
		Finally, the last equality holds because, by Lemma~\ref{lem: BallIntersect}, we know that any 
		candidate $w_2\in B(z_1,w_1)\cap B(z_2,w_1)$ is selected with consensus voting. This is 
		equivalent to $l(z_1,z_2)\subseteq H_2(w_1,w_2)$, by Theorem~\ref{thm:voronoi}. 
	\end{proof}

%% file: lower.tex
\section{Lower bounds on Expected Group Size}\label{sec: LB}
In this section we provide lower bounds for the cases where $K$ is either a unit ball~$\mathbb{B}$ 
or a unit square~$\mathbb{H}$.
%
Recall, we assumed that $\vol(K)=1$ but $\vol(\mathbb{B})=\pi$ for the unit ball. 
We remark that this is of no consequence as we may absorb the associated constant 
factor into our bounds. 
 
\subsection{Lower Bound for the Unit Ball}\label{sec: LBunitBall}

For the unit ball~$\mathbb{B}$, a {\em circular segment} is the small piece of the circle formed by cutting along a {\em chord}.
Evidently, this means that every cap of~$\mathbb{B}$ is either a circular segment or the complement of a circular segment.
Thus, to analyze the case of the unit ball we must study circular segments.

\begin{lemma}\label{lem: BallCap}
	Let $\mathbb{B}$ be a unit ball, and  $J_\delta$ be a circular segment with height $\delta\leq \frac{1}{8}$ then 
	$$\Pr[Z(J_\delta) \ | \ \overbar{J_\delta} \ ]  \gtrsim  \delta^4$$
\end{lemma}
\begin{proof}
	Let $z_1$ and $z_2$ be the endpoints of the chord defining the line segment $J_\delta$.
	By rotating the ball, we may assume the chord is horizontal. Furthermore, by translating the coordinates,
	we may assume that $z_1$ lies at the origin. This is illustrated in Figure~\ref{fig: BallCap}.
	
	\begin{figure}[h!]
		\centering
		\includegraphics[width=.50\textwidth]{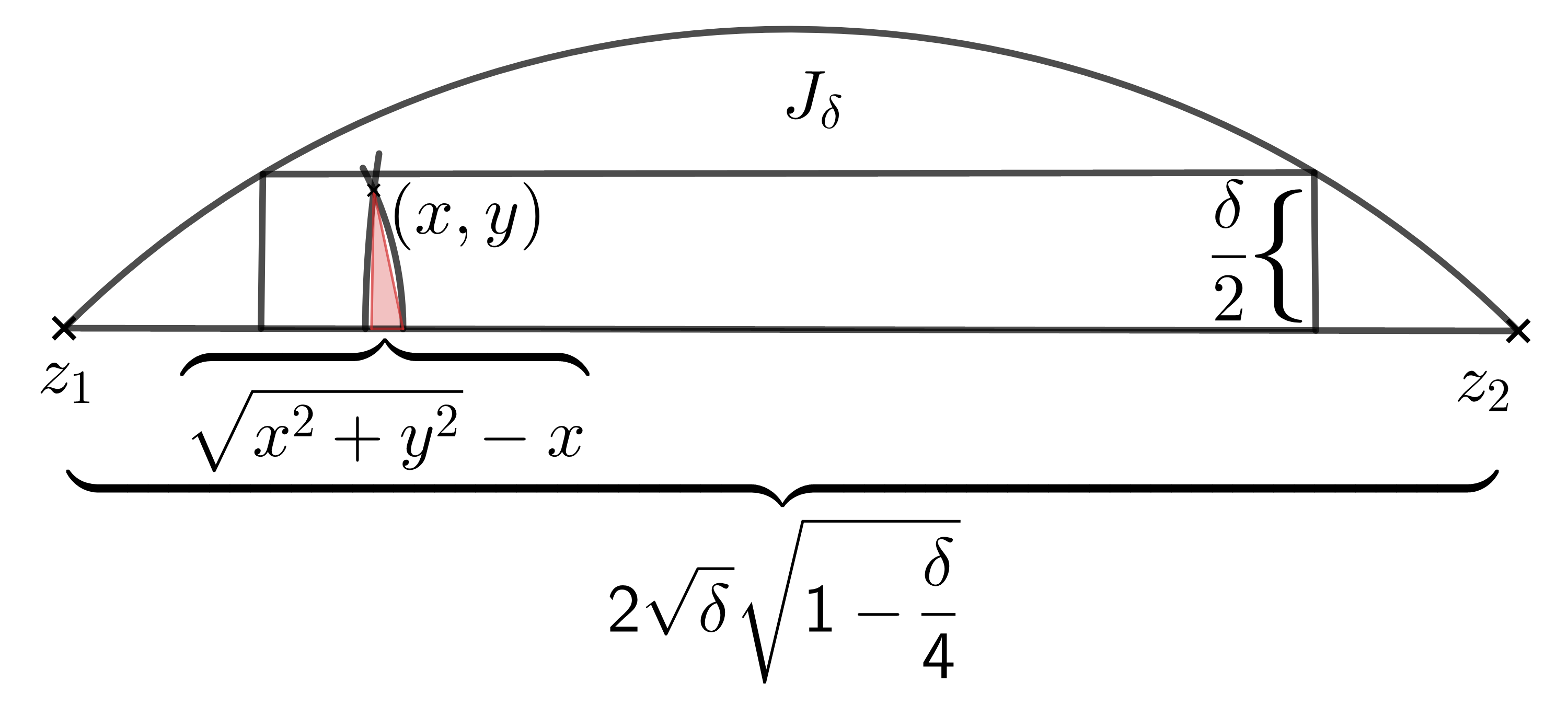}
		\caption{The circular segment $J_\delta$ with inscribed rectangle $R$ of height $\frac{\delta}{2}$}
		\label{fig: BallCap}
	\end{figure}
	
	To apply Lemma~\ref{lem: BndCap}, let $\xi=(x,y)$. Observe that the region
	$B(z_1,\xi)\cap B(z_2,\xi) \cap J_\delta$ contains a triangle $T$
	of height $y$ and width $\sqrt{x^2+y^2}-x$. Again, this is shown in Figure~\ref{fig: BallCap}.
	It follows that
	\begin{equation}\label{eq:BBJ}
	\vol(B(z_1,\xi)\cap B(z_2,\xi) \cap J_\delta) \ \geq\ \frac{1}{2}y\cdot (\sqrt{x^2+y^2}-x) \ \geq \ \frac{y^3}{4(x+y)}
	\end{equation}
	To obtain inequality~(\ref{eq:BBJ}) first observe,  by {\em Bernoulli's Inequality}, that $\sqrt{1+\frac{y^2}{x^2}}\leq 1+ \frac{1}{2} \frac{y^2}{x^2}$. 
	This implies $\left(1+\frac{y^2}{x^2}\right)-\sqrt{1+\frac{y^2}{x^2}}\geq\frac{1}{2}\frac{y^2}{x^2}$.
	Rearranging terms, we obtain $(x^2+y^2)-x\sqrt{x^2+y^2}\geq \frac{y^2}{2}$. 
	Equivalently, $\sqrt{x^2+y^2}-x\geq \frac{y^2}{2\sqrt{x^2+y^2}}\geq \frac{y^2}{2(x+y)}$, as required.
	Combining Lemma~\ref{lem: BndCap} with inequality~(\ref{eq:BBJ}) then gives	
	\begin{equation}\label{eq:prob-J}
	\Pr[Z(J_\delta) \ | \ \overbar{J_\delta} \ ] \ \geq\ \int_{J_\delta} \, \vol(B(z_1,\xi)\cap B(z_2,\xi) \cap J_\delta) \ d\xi 
	\ \geq\  \int_{J_\delta} \frac{y^3}{4(x+y)} \ dx\, dy 	
	\end{equation}
	To lower bound this integral, rather than integrate over the entire circular segment $J_\delta$, we will integrate over an inscribe
	rectangle. Specifically let $R$ be the rectangle of height $\frac12 \delta$ inscribed inside~$J_\delta$.
	This rectangle, shown in Figure~\ref{fig: BallCap}, has width $2\sqrt{1-(1-\frac{\delta}{2})^2}=2\sqrt{\delta}\cdot \sqrt{1-\frac{\delta}{4}}$
	and is centred at $x= \sqrt{\delta}\cdot \sqrt{2-\delta}$.
	Therefore,
	\begin{align}\label{eq:y-cubed}
	\int_{J_\delta}\, \frac{y^3}{4(x+y)} \ dx \, dy 
	&\ \ge\  \int_{R}\, \frac{y^3}{4(x+y)} \ dx \, dy  \nonumber \\
	&\ =\  \int_0^{\frac{\delta}{2}} \int^{\sqrt{\delta}(\sqrt{2-\delta}
		+\sqrt{1-\frac{\delta}{4}})}_{\sqrt{\delta}(\sqrt{2-\delta}-\sqrt{1-\frac{\delta}{4}})} \frac{y^3}{4(x+y)} \ dx\, dy \nonumber\\
	&\ \geq\ \frac{1}{8} \cdot \int_0^{\frac{\delta}{2}} \int^{\sqrt{\delta}(\sqrt{2-\delta}
		+\sqrt{1-\frac{\delta}{4}})}_{\sqrt{\delta}(\sqrt{2-\delta}-\sqrt{1-\frac{\delta}{4}})} \frac{y^3}{x} \ dx\, dy \nonumber\\ 
	&\ \gtrsim\  \int_0^{\frac{\delta}{2}} y^3\cdot \Big[\ln x\Big]^{\sqrt{\delta}(\sqrt{2-\delta}
		+\sqrt{1-\frac{\delta}{4}})}_{\sqrt{\delta}(\sqrt{2-\delta}-\sqrt{1-\frac{\delta}{4}})} \ dy \nonumber\\ 
	&\ \gtrsim\  \delta^4 \cdot \ln\left(\frac{\sqrt{2-\delta}+\sqrt{1-\frac{\delta}{4}}}{\sqrt{2-\delta}-\sqrt{1-\frac{\delta}{4}}}\right) \nonumber\\
	&\ \gtrsim\ \delta^4
	\end{align}
	The second inequality holds because $\delta\leq \frac{1}{8}$; specifically, when $\delta\leq \frac{1}{8}$ the circular segment 
	lies under the line $y=x$. Thus, every point inside the circular segment satisfies $y\leq x$. The last inequality follows by observing that the logarithmic 
	term is lower bounded by $\ln\left(\frac{\sqrt{2-\frac{1}{8}}+\sqrt{1-\frac{1}{32}}}{\sqrt{2}-\sqrt{1-\frac{1}{32}}}\right)\geq \ln 5$.
	Combining the inequalities~(\ref{eq:prob-J})  and~(\ref{eq:y-cubed}) completes the proof.
\end{proof}

Note that Lemma~\ref{lem: BallCap} can be used to prove  lower bound on the expected cardinality of the group. It is also used in later sections to obtain appropriate function $f_K$ when using Corollary~\ref{cor: MainBound}.
\begin{theorem} 
	 For the unit ball~$\mathbb{B}$, the expected group size   after $T$ rounds is  $\E[|G^T|]\gtrsim T^{\frac{1}{8}}$.
\end{theorem}
\begin{proof}
	For each $t$, we construct collection $\{A_1^t, A_2^t,\dots, A_{N(t)}^t\}$ of {\em disjoint} circular segments  on the unit ball. To do this, let the height of each circular segment in the collection be $\delta(t)=\frac{1}{4t^{\frac{1}{4}}}$.Then we can fit $N(t)=\left\lfloor \pi\cdot  t^{\frac{1}{8}}\right\rfloor$ of these segments into $\mathbb{B}$.
	To see this, observe that a circular segment of height $\delta$ has a chord of length $2\sqrt{\delta}\sqrt{2-\delta}$. The central  angle of the segment is then $\theta=2\arctan\left(\frac{\sqrt{\delta}\sqrt{2-\delta}}{1-\delta} \right)\leq 4\sqrt{\delta}$, implying 
	the existence
	of at least $N=\left\lfloor\frac{\pi}{2\sqrt{\delta}}\right\rfloor$ disjoint circular segments of height $\delta$. 
	Now define $\tau\leq T$ to be the last round for which $\Pr[S^{\tau}\cap A_i^{\tau}\neq \emptyset ]\geq \frac{1}{2}$. 
	Thus, 
	\begin{align}\label{eq:N-tau}
	\E[|G^{\tau}|] \geq  \sum_{i=1}^{N(\tau)} \Pr[S^{\tau}\cap A_i^{\tau}\neq \emptyset ]   
	 = N(\tau)\cdot \Pr[S^{\tau}\cap A_1^{\tau}\neq \emptyset ]
	 \ge   \frac12  \left\lfloor \pi \tau^{\frac{1}{8}}\right\rfloor \gtrsim {\tau}^{\frac{1}{8}}
	\end{align}
	Here, the first inequality follows from the observation that if  $S^t\cap A_i^t\neq\emptyset$ then there is a least one 
	group member inside $A_i^t$.
	The equality is due to symmetry; that is, $\Pr[S^t\cap A_i^t\neq\emptyset]=\Pr[S^t\cap A_j^t\neq\emptyset] $ for 
	each pair $1\le i,j\le N(t)$. The second inequality follows because, by definition,
	$\Pr[S^{\tau}\cap A_i^{\tau}\neq \emptyset ]\geq \frac{1}{2}$.
	Next consider rounds $t> \tau$. For these rounds, by definition of $\tau$, we know
	$\Pr[S^{t}\cap A_i^{t}\neq \emptyset ]\leq  \frac{1}{2}$ which implies $\Pr[S^{t}\cap A_i^{t}= \emptyset ]\geq  \frac{1}{2}$. Therefore,
	\begin{align}\label{eq:large-t}
	\Pr[X^{t+1}]
	&\ \gtrsim\   \sum_{i=1}^{N(t)} \Pr \left[Z(A^t_i)\ \wedge\ \left(S^t\cap A^t_i= \emptyset\right)  \right] \nonumber \\
	&\ =\  N(t)\cdot \Pr \left[Z(A^t_i)\ | \ S^t\cap A^t_i= \emptyset  \right]\cdot  \Pr[\ S^t\cap A^t_i= \emptyset \ ] \nonumber \\
	&\ \gtrsim \  N(t)\cdot\Pr\left[Z(A^t_i)\ | \ \overbar{A^t_i}\right] \nonumber
	\end{align}
	Where the last inequality follows by Corollary~\ref{cor:subset}.
	Finally by Lemma~\ref{lem: BallCap}, we see $\Pr\left[Z(A^t_i)\ | \ \overbar{A^t_i}\right]\gtrsim \frac{1}{t} $, and thus  $\Pr[X^{t+1}]\gtrsim t^{-\frac{7}{8}}$ for any $t>\tau$.
	We may now lower bound the expected group size at the end of round $T$. Specifically, for $T\geq 4$,
	\begin{align*}
	\E[|G^{T}|]
	 = \E[|G^{\tau}|]+\sum_{t=\tau+1}^{T}\Pr[X^{t}=1] 
	 \gtrsim  {\tau}^{\frac{1}{8}}+\sum_{t=\tau+1}^{T} t^{-\frac{7}{8}} 
	 \gtrsim  T^{\frac{1}{8}} 
	\end{align*}
The last inequality was obtained using integral bounds.
\end{proof}

\subsection{Lower bound for the Unit Square}\label{sec: LBunitSquare}
For the unit square~$\mathbb{H}$ caps are either {\em right-angled triangles} 
or {\em right-angled trapezoids} (trapezoids with 
two adjacent right angles).  
We can bound the probability of accepting a point inside a right-angled trapezoid by consideration of
the largest inscribed triangle it contains. Thus, it suffices to consider only the case in which the cap forms a triangle.
\begin{lemma}\label{lem: f_sq}
	Let $J_{a,b}$ be triangular cap on the unit square with perpendicular side lengths $a\leq b$. Then
	$$
	\Pr[Z( J_{a,b})\ |\ \overbar{J_{a,b}}\  ]\geq \frac{1}{2^{11}}\, a^4\cdot \left(1+\ln\left(\frac{b}{a}\right)\right).
	$$
\end{lemma}
\begin{proof}
	Take the unit square~$\mathbb{H}$ and the cap $J_{a,b}$.
	Without loss of generality, we may assume the hypotenuse is horizontal with one endpoint $z_1$ at the origin.
	This is shown in Figure~\ref{fig: SQ}.
	\begin{figure}[h!]	
		\centering
		\includegraphics[width=.40\textwidth]{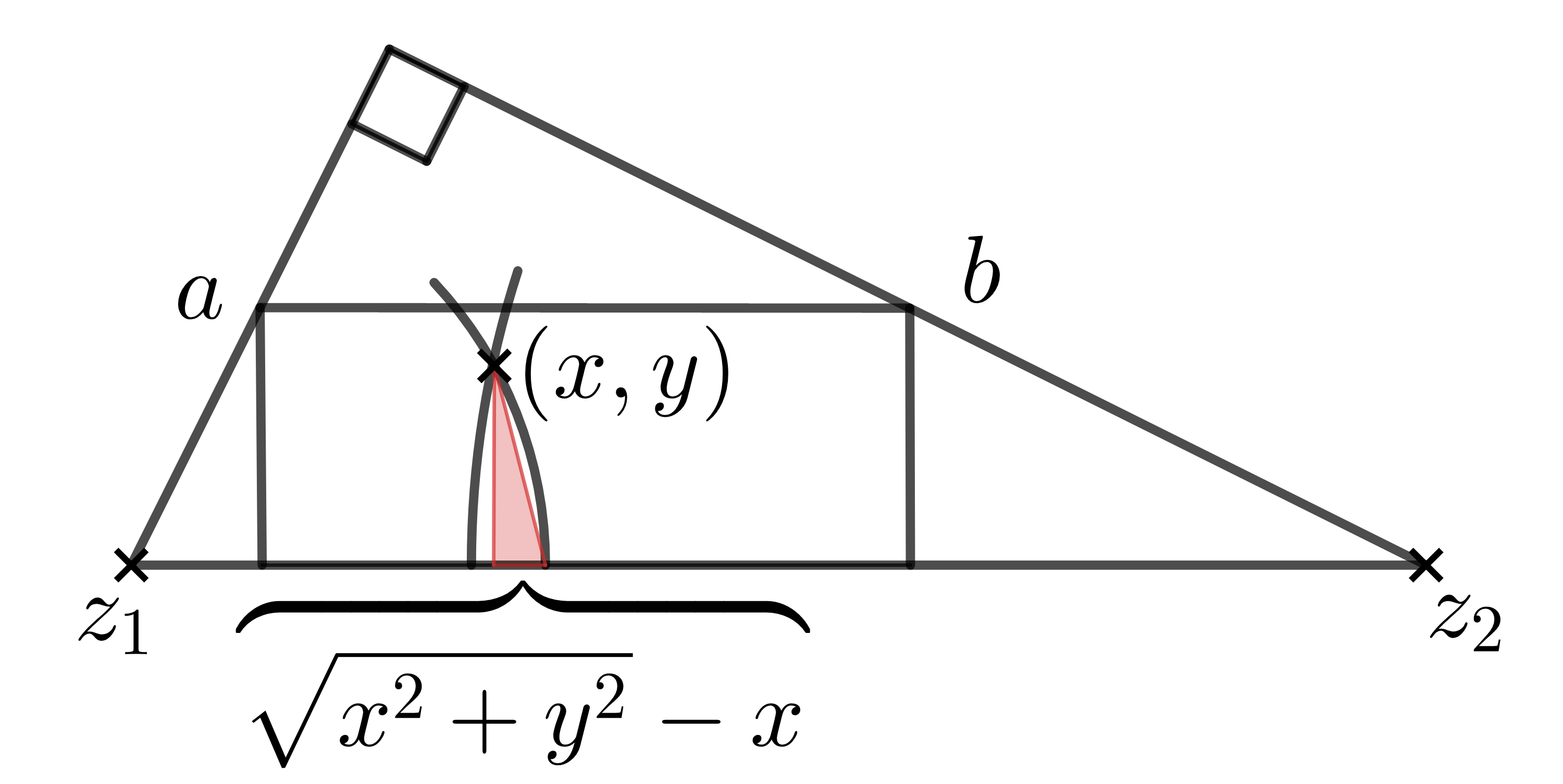}	
		\caption{The Cap $J_{a,b}$ (with sides lengths $a,b$ and height $\frac{ab}{\sqrt{a^2+b^2}}$) and inscribed 
			rectangle $R$ of half the height.}\label{fig: SQ}
	\end{figure}	
	
	Again, setting $\xi=(x,y)$, we have that 
	\begin{equation}\label{eq:BBJ-2}
	\vol(B(z_1,\xi)\cap B(z_2,\xi) \cap J_\delta) \ \geq\ \frac{1}{2}y\cdot (\sqrt{x^2+y^2}-x) \ \geq \ \frac{y^3}{4(x+y)}
	\end{equation}
	Combining Lemma~\ref{lem: BndCap} with inequality~(\ref{eq:BBJ-2}) gives	
	\begin{equation}\label{eq:prob-J-2}
	\Pr[Z(J_{a,b}) \ | \ \overbar{J_{a,b}}\ ] \ \geq\  \int_{J_{a,b}} \vol(B(z_1,\xi)\cap B(z_2,\xi)\cap J_{a,b} )\ d\xi 
	\ \geq\ \int_{J_{a,b}} \frac{y^3}{4(x+y)}\ dx \,dy \\
	\end{equation}
	Recall the cap $J_{a,b}$ has perpendicular side lengths $a \le b$. Thus the height of $J_{a,b}$ is $\frac{ab}{\sqrt{a^2+b^2}}$.
	Again, a lower bound can be obtained by integrating over an inscribed rectangle rather than over the entire cap $J_{a,b}$.
	Specifically, let $R$ be the inscribed rectangle with half the height of $J_{a,b}$; see Figure~\ref{fig: SQ}. 
	Then
	\begin{align}\label{eq:y-cubed-2}
	\int_{J_{a,b}}\, \frac{y^3}{4(x+y)}\ dx\, dy &\ \ge\  \int_{R}\, \frac{y^3}{4(x+y)}\  dx\,  dy  \nonumber \\
	&\ \geq\ \int_{0}^{\frac{ab}{2\sqrt{a^2+b^2}}}\int_{\frac{a^2}{2\sqrt{a^2+b^2}}}^{\frac{a^2}{\sqrt{a^2+b^2}}
		+\frac{b^2}{2\sqrt{a^2+b^2}}}  \frac{y^3}{4(x+y)} \ dx\, dy \nonumber\\ 
	&\ =\ \frac{1}{4}\cdot  \int_{0}^{\frac{ab}{2\sqrt{a^2+b^2}}} y^3\cdot \Big[\ln(x+y)\Big]_{\frac{a^2}{2\sqrt{a^2+b^2}}}^{\frac{a^2}{\sqrt{a^2+b^2}}} \ dy \nonumber\\
	&\ =\ \frac{1}{4}\cdot  \int_{0}^{\frac{ab}{2\sqrt{a^2+b^2}}} y^3\cdot \ln\left(\frac{\frac{a^2}{\sqrt{a^2+b^2}}
		+\frac{b^2}{2\sqrt{a^2+b^2}}+y}{\frac{a^2}{2\sqrt{a^2+b^2}}+y} \right) \, dy \nonumber\\
	&\ =\ \frac{1}{4} \cdot \int_{0}^{\frac{ab}{2\sqrt{a^2+b^2}}} y^3\cdot 
	\ln\left(1+\frac{a^2+b^2}{a^2+2y\sqrt{a^2+b^2}}\right)\, dy \nonumber \\
	&\ \ge\ \frac{1}{4} \cdot \int_{0}^{\frac{ab}{2\sqrt{a^2+b^2}}} y^3\cdot 
	\ln\left(1+\frac{a^2+b^2}{a^2+ab}\right)\, dy
	\end{align}
	To simplify (\ref{eq:y-cubed-2}) recall that $a\le b$, by assumption. Thus $\frac{b}{\sqrt{a^2+b^2}}\geq \frac{1}{\sqrt{2}}\ge \frac{1}{2}$.	 
	This implies that
	\begin{align}\label{eq:y-cubed-3}
	\int_{J_{a,b}}\, \frac{y^3}{4(x+y)}\,dx\, dy 	
	&\ \geq\  \frac{1}{4}\cdot \int_{0}^{\frac{ab}{2\sqrt{a^2+b^2}}} y^3\cdot \ln\left(1+\frac{a^2+b^2}{a^2+ab}\right) \, dy\nonumber\\ 
	&\ \geq\  \frac{1}{8}\cdot  \left(1+\ln\left(\frac{b}{a}\right)\right)\cdot \int_{0}^{\frac{ab}{2\sqrt{a^2+b^2}}} y^3 \, dy \nonumber\\
	&\ \geq\  \frac{1}{8}\cdot  \left(1+\ln\left(\frac{b}{a}\right)\right)\cdot \int_{0}^{\frac{a}{4}} y^3 \, dy \nonumber\\
	&\ =\  \frac{1}{2^{11}} \, a^4\cdot \left(1+\ln\left(\frac{b}{a}\right)\right)
	\end{align}
	Here the second inequality arises since 
	$$\left(1+\frac{a^2+b^2}{a^2+ab}\right)^2\ =\ \left(\frac{b}{a}+\frac{2}{\frac{b}{a}+1}\right)^2
	\ \geq\  \left(\frac{b}{a}+\frac{a}{b}\right)^2\ \geq\ 2+\left(\frac{b}{a}\right)^2 \ \geq\  e\cdot \frac{b}{a}$$
	Thus, taking the log on both sides gives the second inequality. Putting together 
	inequalities~(\ref{eq:prob-J-2}) and~(\ref{eq:y-cubed-3}) completes the proof.
\end{proof}

\begin{theorem}
	For the unit square $\mathbb{H}$, the expected group size  after $T$ rounds is $\E[|G^{T}|] \gtrsim\ln(T)$.
\end{theorem}
\begin{proof}
	For each $t$, consider the triangle $A^t=\conv \left((0,0),\left(0,\frac{1}{4t^{\frac{1}{4}}}\right),\left(\frac{1}{4t^{\frac{1}{4}}},0\right)\right)$.
	As discussed, the triangle $A^t$ is a cap of the unit square. Observe that,
	\begin{align}\label{eq: A^t}
	\Pr[X^{t+1}]
	&\ \geq \ \Pr[Z(A^t) \wedge (S^t\cap A^t=\emptyset) ] \nonumber \\ 
	&\ =\ \Pr[Z(A^t)\ |\ (S^t\cap A^t=\emptyset)  \ ] \cdot  \Pr[\ S^t\cap A^t=\emptyset \ ]  \nonumber \\ 
	&\ \ge \ \Pr[Z(A^t)\ |\ \overbar{A^t}\ ] \cdot  \Pr[\ S^t\cap A^t=\emptyset \ ]  \nonumber \\
	&\ \gtrsim\ \frac{1}{t}\cdot \Pr[\ S^t\cap A^t=\emptyset \ ]
	\end{align} 
	Here the second inequality follows from Corollary~\ref{cor:subset}. The third inequality 
	is derived by applying Lemma~\ref{lem: f_sq} with respect to the cap $A^t$, for which $a=b= \frac{1}{4t^{\frac{1}{4}}}$.
	\begin{claim}\label{claim: Q}
		There exist constants $c,\tau$ such that $\Pr[\ S^t\cap A^t=\emptyset \ ]\geq c $ for all $t\geq \tau$
	\end{claim}
	\begin{proof}[proof of claim]
		First we show that there exists constant $\tau$ such that the convex hull $S^\tau$ contains the square $Q=\left[\frac{1}{4},\frac{3}{4}\right]^2$ with probability bounded below by some positive constant. To see this, let $T_1,T_2,T_3$ and $T_4$ be right-angled triangles each containing one of the corners with perpendicular side lengths
		$a=b=\frac14$. By Lemma~\ref{lem: f_sq}, the probability of accepting a candidate inside triangle $T_\ell$, given
		no member is currently in $T_\ell$, is lower bounded by  $\frac{a^4}{2^{11}}=\frac{1}{2^{19}}$. Hence we have,
		\begin{equation}\label{eq:T-ell}
		\Pr[T_\ell \cap S^{i} = \emptyset]  \ \leq\  \left(1-\frac{1}{2^{19}}\right)^i \ = \ k^i
		\end{equation}
			\begin{figure}[ht]	
			\centering
			\includegraphics[width=.42\textwidth]{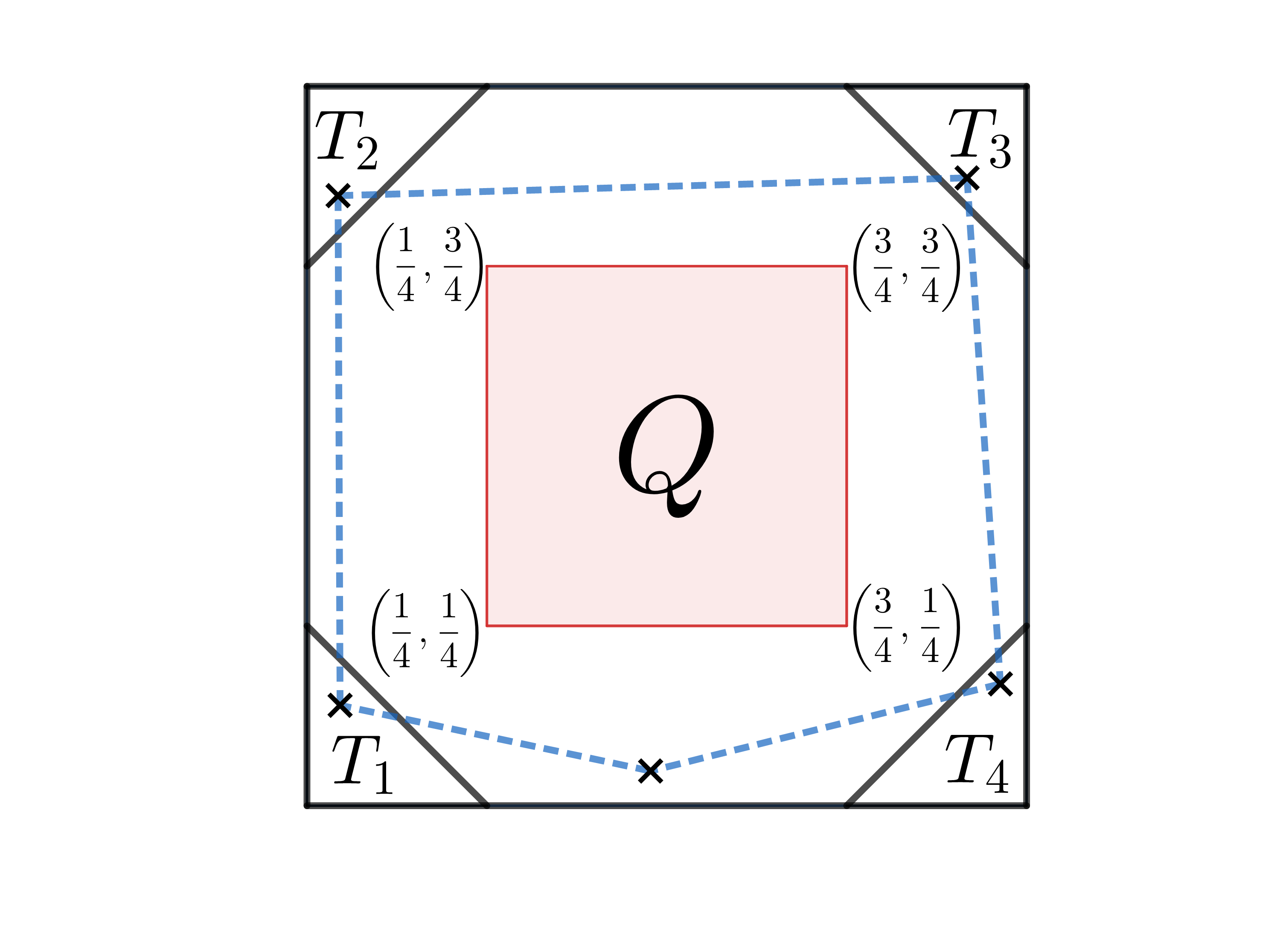}
			\includegraphics[width=.42\textwidth]{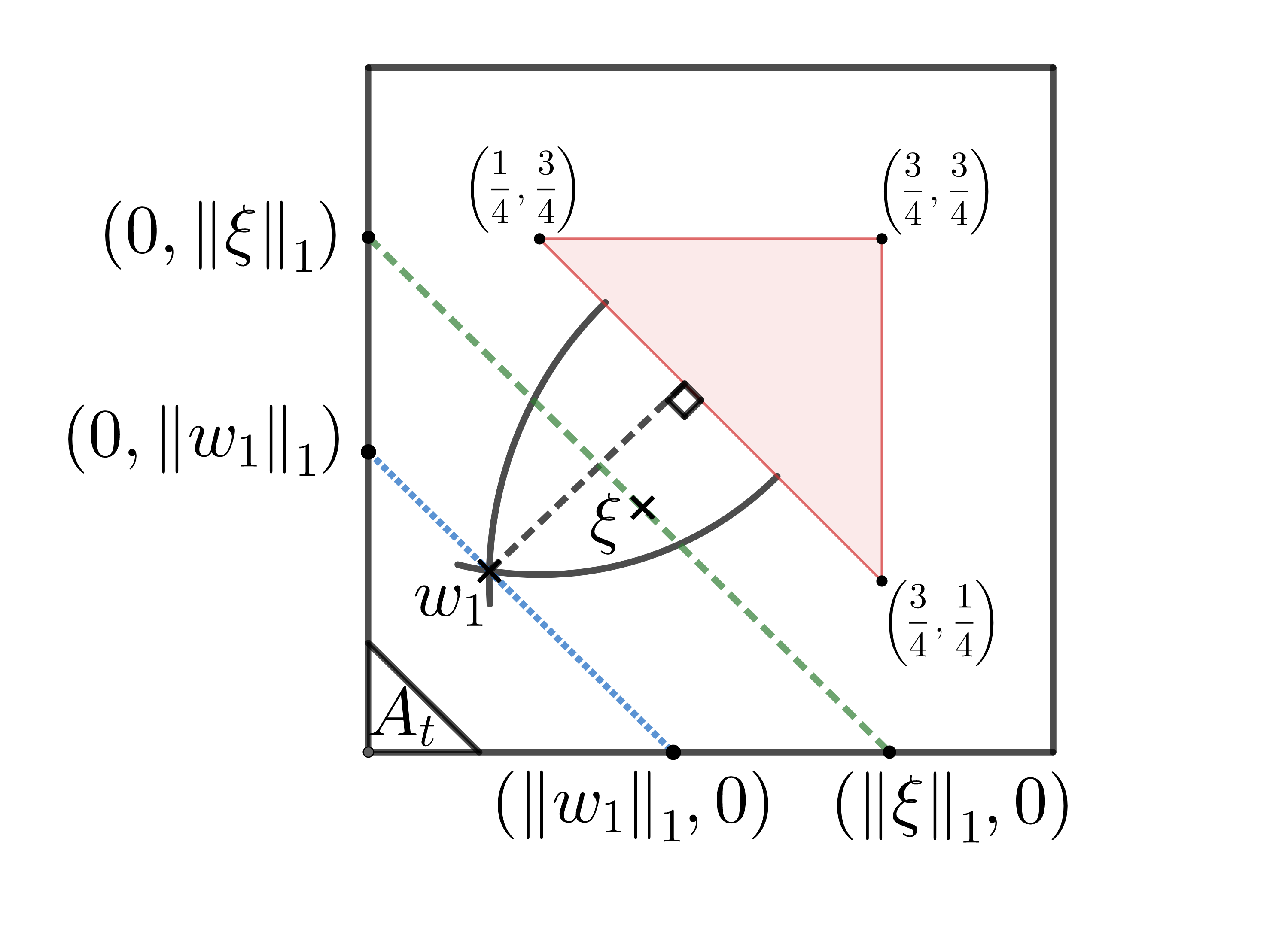}
			\caption{ Figure on the left illustrates that if there is a member in each $T_\ell$ then $Q\subseteq  S_i$. The figure on the right 
				shows if $\xi \in B\left(\left(\frac{1}{4},\frac{3}{4}\right),w_1\right)\cap B\left(\left(\frac{3}{4},\frac{1}{4}\right),w_1\right) $ 
				then $\| \xi \|_1\geq \|w_1\|_1$  } \label{fig:LowerBoundSquare}
		\end{figure}
		
			Observe that if there is a member of $S^i$ selected in each of the four triangles $T_1,T_2,T_3$ and $T_4$
		then $Q\subseteq S^i$, as illustrated in Figure~\ref{fig:LowerBoundSquare}.  Consequently, applying the union 
		bound and (\ref{eq:T-ell}), for all $1\le \ell \le 4$, we have $\Pr[Q\not\subseteq S^i] \leq 4k^i$.  Note, since $k=1-\frac{1}{2^{19}}$, there exists 
		a fixed constant $\tau$ such that $4k^\tau <1$; hence, $\Pr[Q \subseteq S^\tau]\geq 1-4k^\tau>0 $. Consider any round $t\geq \tau+1$ then,
		\begin{align}
		\Pr[\ S^t\cap A^t=\emptyset \ ] &\geq \Pr[\ (S^t\cap A^t=\emptyset) \wedge (Q\subseteq S^\tau ) \ ] \nonumber  \\
		&=\Pr[\ S^t\cap A^t=\emptyset\ |\ Q\subseteq S^\tau  \ ] \Pr[Q\subseteq S^\tau] \label{eq: QL}
		\end{align}
		Since we know that $\Pr[Q\subseteq S^\tau]$ is bounded below by a constant, we need to show the remaining term of (\ref{eq: QL}) is bounded below by constant. Recursively we have, 
		\begin{align}
		\Pr[ S^t\cap A^t&=\emptyset \ | \ Q\subseteq S^\tau \ ] \nonumber\\
		&\ = \ \Pr[\neg Z(A^t) \wedge(S^{t-1}\cap A^t=\emptyset)\  |\  Q\subseteq S^\tau] \nonumber \\
		&\ =\ \Pr[\neg Z(A^t) \ | \ (S^{t-1}\cap A^t=\emptyset)\wedge  (Q\subseteq S^\tau) ] \cdot  \Pr[\ S^{t-1}\cap A^t=\emptyset \ |Q\subseteq S^\tau ]   \nonumber \\
		&\ =\ \Pr[\ S^\tau \cap A^t=\emptyset\ | \ Q\subseteq S^\tau ] \cdot \prod_{i=\tau}^{t-1} \, (1-\Pr[Z(A^t) \ | \ S^{i}\cap A^t=\emptyset \wedge (Q\subseteq S^\tau ) ])   \nonumber \\
		&\ \gtrsim\   \Pr[\ S^\tau \cap A^\tau =\emptyset\ | \ Q\subseteq S^\tau ] \cdot\prod_{i=\tau}^{t-1} \, (1-\Pr[Z(A^t) \ | \ S^{i}\cap A^t=\emptyset \wedge (Q\subseteq S^\tau ) ]) \nonumber \\
		&\ \gtrsim\   \Pr[\ S^\tau \cap A^\tau =\emptyset\ | \ Q\subseteq S^\tau ] \cdot\prod_{i=\tau}^{t-1} \, (1-\Pr[Z(A^t) \ | \ Q ]) \nonumber \\
		&\ \gtrsim\ \prod_{i=\tau}^{t-1} \, (1-\Pr[Z(A^t) \ | \ Q ]) \label{eq: QS}
		\end{align}
		Here the first inequality holds because, by definition, $A^t\subseteq A^\tau $ for all $t\geq \tau+1$. For  the second inequality, observe that  $S^\tau \subseteq S^i$ when $i\geq \tau$.
		Thus, by Corollary~\ref{cor:subset}, $\Pr[Z(A^t) \ | \ S^{i}\cap A^t=\emptyset \wedge (Q\subseteq S^\tau ) ]\leq\Pr[Z(A^t) \ |\  Q\ ]$. The last inequality holds because $\tau$ is constant, so  $A^\tau$ is a small fixed triangle and  there is a positive probability that no candidate was selected inside $A^\tau$ in the first $\tau$ rounds, even when the convex hull of the members includes the square $Q$. 
		
		Next, we claim that if $Q$ is the convex hull then a candidate can be accepted inside $A^t$ only if
		both candidates are in $A^t$. In particular, this gives the following useful inequality:
		\begin{equation}\label{eq:useful}
		\Pr[\ Z(A^t)\ | \ Q] \ \leq\  \vol(A^t)^2 \ =\ \left(\frac12\left(\frac{1}{4t^{\frac14}}\right)^2\right)^2\ =\ \frac{1}{2^{10}\cdot t}
		\end{equation}
		To see this, suppose the claim is false. That is, $w_2\in A^t$ is selected when $Q$ is the convex hull but $w_1\notin A^t$. 
		Then, by Lemma~\ref{lem: BallIntersect}, it must  be the case that
		$$w_2 \in  B\left(\left(\frac{1}{4},\frac{1}{4}\right),w_1\right) \cap B\left(\left(\frac{3}{4},\frac{1}{4}\right),w_1\right)
		\cap  B\left(\left(\frac{1}{4},\frac{3}{4}\right),w_1\right) \cap  B\left(\left(\frac{3}{4},\frac{3}{4}\right),w_1\right)$$
		Observe that $w_1\in A^t$ if and only if  $\|w_1\|_1=|(w_1)_1|+|(w_1)_2|\leq \frac{1}{4t^{\frac{1}{4}}}$, where $(w_1)_i$ 
		denotes the $i$'th component of $w_1$. Since $w_1\not\in A^t$, 
		we have $\|w_1\|_1>\frac{1}{4t^{\frac{1}{4}}}$. As illustrated in Figure~\ref{fig:LowerBoundSquare}, the region
		$B((\frac{1}{4},\frac{3}{4}),w_1)\cap B((\frac{3}{4},\frac{1}{4}),w_1)$ does not 
		intersect $A^t$ if $w_1\not\in A^t$. Thus, the winner $w_2$ cannot be inside $A^t$ and the claim is verified.
		
		Finally combining (\ref{eq: QL}),(\ref{eq: QS}) and (\ref{eq:useful}), we have for any $t\geq \tau+1$ 
		\begin{align*}
		\Pr[ S^t\cap A^t=\emptyset]\gtrsim \prod_{i=\tau}^{t-1}\left(1- \frac{1}{2^{10} \cdot t}\right) 
		&\ \gtrsim\   \left(1- \frac{1}{2^{10} \cdot t}\right)^{t}\qedhere
		\end{align*}
	\end{proof}
	Using the Claim~\ref{claim: Q} and (\ref{eq: A^t}) we get  $\Pr[X^{t+1}]\gtrsim \frac{1}{t}$ for any $t\geq \tau+1$. Ergo, we have\, $\E[G^T]  =  \sum_{t=1}^T \Pr[X^t]  \gtrsim  \sum_{t=\tau+1}^T \frac{1}{t}  \cong  \ln(T) $.
\end{proof}

%% file: upper.tex
\section{Upper Bounds on Expected Group Size}\label{sec: UB}
We now apply the techniques developed in Sections~\ref{sec: MainTools} and \ref{sec: LB} to upper bound the 
expected cardinality of the group for the unit ball and the unit square. Specifically, we apply Corollary~\ref{cor: MainBound} 
to these metric spaces 
using the $f_K$ obtained from Lemma~\ref{lem: BallCap} and Lemma~\ref{lem: f_sq}, respectively.

\subsection{Upper Bound for the Unit Ball}
Observe that exactly one of the two Voronoi regions corresponds to a circular segment. Furthermore, since a circular 
segment fits inside its complement, $\arg\min\limits_i \Pr[Z(H_i(w_1,w_2))|\overbar{H_i(w_1,w_2)}]$ is attained by 
the $H_i$ corresponding to a circular segment. Let $\delta(w_1,w_2)$ denote the height of the circular segment for this 
Voronoi region $H_i$. Then, by Lemma~\ref{lem: BallCap}, we 
have $\min\limits_{i} \Pr[\ Z(H_i(w_1,w_2))\ |\ \overbar{H_i(w_1,w_2)}\ ]\gtrsim\delta(w_1,w_2)^4$. 
Thus  $f_\mathbb{B}(w_1,w_2)=\delta(w_1,w_2)^4$, satisfies the conditions of Corollary~\ref{cor: MainBound}. 
However when using Corollary~\ref{cor: MainBound} we need to understand $\Phi(\lambda)$; Lemma~\ref{lem: UBBall} allows us to do exactly that. 

\begin{lemma}\label{lem: UBBall}
	Let $\delta(w_1,w_2)$ be the height of the circular segment $H_i(w_1,w_2)$ formed by the Voronoi regions. 
	Then, for all $\lambda\leq \frac{1}{10^4}$, we have
	$$\Phi(\lambda) \ =\ \int \int \mathbb{I}\left[\delta(w_1,w_2)^4\leq \lambda\right] \, dw_1 \, dw_2 \ \lesssim\  \lambda^{\frac{7}{8}}$$
\end{lemma}
\begin{proof}[Proof of Lemma~\ref{lem: UBBall}]
	Let $\mathcal{H}$ be the hyperplane separating the two Voronoi regions. Since $\mathcal{H}$ 
	contains the midpoint $\frac12 (w_2+w_1)$ and has normal vector $w_2-w_1$, it holds that
	\begin{align}\label{eq:hyperplane}
	\mathcal{H}(w_1,w_2) &= \Big\{ \xi\in \mathbb{R}^2 \ : \ (w_2-w_1)\cdot \left( \xi - \frac{w_2+w_1}{2} \right) = 0 \Big\}
	\end{align}
	By basic algebra, the distance from the origin to $\mathcal{H}$ is 
	$\frac{|\|w_2\|^2-\|w_1\|^2|}{2\sqrt{(w_2-w_1)^2_2+(w_2-w_1)^2_1}}$.
	Hence we have,	
	\begin{align}\label{eq: delta}
	\delta(w_1,w_2)=1- \left(\frac{| \|w_2\|^2-\|w_1\|^2|}{2\sqrt{(w_2-w_1)^2_2+(w_2-w_1)^2_1}} \right) 
	\end{align}
	Note that $\delta(w_1,w_2)=\delta(w_2,w_1)$ because swapping positions of the candidates does not change the sizes of Voronoi regions. 
	It will now be convenient to work with polar coordinates. So denote $w_1=(r,\varphi)$ and $w_2=(\xi,\theta)$. Observe that  
	\begin{align}\label{eq: BPhi1}
	\Phi(\lambda)
	&\ =\ \int \int \mathbb{I}\left[\delta(w_1,w_2)\leq \lambda^{\frac{1}{4}}\right] \, dw_2\,  dw_1 \nonumber \\
	&\ =\ 2\int \int \mathbb{I}\left[\delta(w_1,w_2)\leq \lambda^{\frac{1}{4}},\|w_1\|\leq\|w_2\| \right] \, dw_2\,  dw_1 \nonumber \\ 
	&\ =\ 2 \int_{-\pi}^{\pi} \int_0^1 \int \mathbb{I}\left[\delta((r,\varphi),w_2)\leq \lambda^{\frac{1}{4}},r\leq\|w_2\|\right] \, dw_2 \, r\, dr \, d\varphi  \nonumber \\ 
	&\ =\ 4\pi \int_0^1 \int  \mathbb{I}\left[\delta((r,0),w_2)\leq \lambda^{\frac{1}{4}}, r\leq \|w_2\|\right] \, dw_2 \, r \, dr \nonumber  
		\end{align}
	\begin{align}
	&\ =\ 4\pi \int_0^1  \int_r^1 \int_{-\pi}^\pi \mathbb{I}\left[\delta((r,0),(\xi ,\theta))\leq \lambda^{\frac{1}{4}}\right] \, \xi \, d\theta\,  d\xi   \, r\,  dr   \nonumber \\ 
	&\ \leq\ 4\pi \int_0^1  \int_r^1 \int_{-\pi}^\pi \mathbb{I}\left[\delta((r,0),(1 ,\theta))\leq \lambda^{\frac{1}{4}}\right]  \, \xi \,  d\theta \,  d\xi \, r \, dr \nonumber \\
	&\ =\ 4\pi \int_0^1   \int_{-\pi}^\pi \mathbb{I}\left[\delta((r,0),(1 ,\theta))\leq \lambda^{\frac{1}{4}}\right]\left(\frac{1-r^2}{2}\right) \, r \, d\theta \, dr 
	\end{align}
	The first equality follows as $\delta(w_1,w_2)=\delta(w_2,w_1)$. 
	The fourth equality holds by rotational symmetry. 
	The inequality holds because, by equation~(\ref{eq: delta}), we have 
	\begin{align*}
	\delta((r,0),(\xi ,\theta))
	&\ =\ 1- \left(\frac{\xi^2-r^2}{2} \right) \frac{1}{\sqrt{\xi^2 \sin^2(\theta)+(\xi \cos(\theta)-r)^2}} \\
	&\ =\ 1- \left(\frac{\xi^2-r^2}{2} \right) \frac{1}{\sqrt{\xi^2+r^2-2r\xi\cos(\theta)}} \\
	&\ \geq\ 1- \left(\frac{1-r^2}{2} \right) \frac{1}{\sqrt{1+r^2-2r\cos(\theta)}} \\ 
	&\ =\  \delta((r,0), (1,\theta)) 
	\end{align*}
	Note that if $r<1-2\lambda^\frac{1}{4}$ then we have $\delta((r,0),( 1,\theta ))\geq \delta((r,0),( 1,0 )) =  \frac{1-r}{2}>\lambda^\frac{1}{4}$. Hence we see 
	\begin{align}\label{eq:BPhi2}
	\int_0^1   \int_{-\pi}^\pi \mathbb{I}&\left[\delta((r,0),(1 ,\theta))\leq \lambda^{\frac{1}{4}}\right]\left(\frac{1-r^2}{2}\right) \, r \, d\theta \, dr  \nonumber \\
	&\ =\ \int_{1-2\lambda^\frac{1}{4}}^1   \int_{-\pi}^\pi \mathbb{I}\left[\delta((r,0),(1 ,\theta))\leq \lambda^{\frac{1}{4}}\right]\left(\frac{1-r^2}{2}\right) \, r \, d\theta \, dr \nonumber \\ 
	&\ \lesssim\  \lambda^\frac{1}{2} \max\limits_{r\in [1-2\lambda^\frac{1}{4},1]}\int_{-\pi}^\pi \mathbb{I}\left[\delta((r,0),(1 ,\theta))\leq \lambda^{\frac{1}{4}}\right] d\theta 
	\end{align}
	Here the last inequality hold as $\int_{1-2\lambda^\frac{1}{4}}^1 \left(\frac{1-r^2}{2}\right)r dr \leq \int_{1-2\lambda^\frac{1}{4}}^1 \left(1-r\right) dr =2\lambda^{\frac{1}{2}}$. 
	Finally, note that 
	\begin{align}\label{eq: BPhi3}
	\max\limits_{r\in [1-2\lambda^\frac{1}{4},1]}\  & \int_{-\pi}^\pi \mathbb{I}\left[\delta((r,0),(1,\theta))\leq \lambda^{\frac{1}{4}}\right] \, d\theta \nonumber \\
	&\ =\ \max\limits_{r\in [1-2\lambda^\frac{1}{4},1]}\int_{-\pi}^\pi \mathbb{I}\left[1- \left(\frac{1-r^2}{2} \right) \frac{1}{\sqrt{1+r^2-2r\cos(\theta)}}\leq \lambda^{\frac{1}{4}}\right] \, d\theta \nonumber \\ 
	&\ \leq\ \max\limits_{r\in [1-2\lambda^\frac{1}{4},1]}\int_{-\pi}^\pi \mathbb{I}\left[1- \left(\frac{1-r^2}{2} \right) \frac{1}{\sqrt{1+r^2-2r(1- \frac{\theta^2}{8})}}\leq \lambda^{\frac{1}{4}}\right] \, d\theta \nonumber\\ 
	&\ =\ \max\limits_{r\in [1-2\lambda^\frac{1}{4},1]}\int_{-\pi}^\pi \mathbb{I}\left[\sqrt{1+r^2-2r(1- \frac{\theta^2}{8})} \leq \frac{1-r^2}{2(1-\lambda^{\frac{1}{4}})} \right] \, d\theta \nonumber\\ 
	&\ =\ \max\limits_{r\in [1-2\lambda^\frac{1}{4},1]}\int_{-\pi}^\pi \mathbb{I}\left[\sqrt{(1-r)^2+\frac{r\theta^2}{4}} \leq \frac{1-r^2}{2(1-\lambda^{\frac{1}{4}})} \right] \, d\theta \nonumber\\
	&\ =\ \max\limits_{r\in [1-2\lambda^\frac{1}{4},1]}\int_{-\pi}^\pi \mathbb{I}\left[\sqrt{1+\frac{r\theta^2}{4(1-r)^2}}\leq \frac{1+r}{2(1-\lambda^\frac{1}{4})} \right] \, d\theta 
	\end{align}
	Now, if $|\theta|\geq 8\lambda^{\frac{3}{8}}$ then for any $r\in [1-2\lambda^\frac{1}{4},1]$, we have 
	\begin{align*}
	\sqrt{1+\frac{r\theta^2}{4(1-r)^2}}- \frac{1+r}{2(1-\lambda^\frac{1}{4})}
	&\ \geq\ \sqrt{1+\frac{16r\lambda^\frac{3}{4}}{(1-r)^2}}- \frac{1+r}{2(1-\lambda^\frac{1}{4})} \\
	&\ \geq\ \sqrt{1+\frac{16r\lambda^\frac{3}{4}}{4\lambda^\frac{1}{2}}}- \frac{1+r}{2(1-\lambda^\frac{1}{4})} \\
	&\ \geq\ \sqrt{1+4r\lambda^\frac{1}{4}}- \frac{1}{(1-\lambda^\frac{1}{4})}\\
	&\ \geq\ \sqrt{1+4(1-2\lambda^\frac{1}{4})\lambda^\frac{1}{4}}- \frac{1}{(1-\lambda^\frac{1}{4})}\\ 
	&\ >\  0
	\end{align*}
	Here the last inequality holds for any $\lambda^{\frac{1}{4}}\leq \frac{1}{10}$.  Hence,
	\begin{align}
	\max\limits_{r\in [1-2\lambda^\frac{1}{4},1]}&\int_{-\pi}^\pi \mathbb{I}\left[\sqrt{1+\frac{r\theta^2}{4(1-r)^2}}\leq \frac{1+r}{2(1-\lambda^\frac{1}{4})} \right]d\theta
	\nonumber \\
	&\leq \max\limits_{r\in [1-2\lambda^\frac{1}{4},1]}\int_{-8\lambda^{\frac{3}{8}}}^{8\lambda^{\frac{3}{8}}} \mathbb{I}\left[\sqrt{1+\frac{r\theta^2}{4(1-r)^2}}\leq \frac{1+r}{2(1-\lambda^\frac{1}{4})} \right]d\theta \nonumber \\
	&\leq 16 \lambda^{\frac{3}{8}} \label{eq: BPhi4}
	\end{align}
	Finally, combining (\ref{eq: BPhi1}), (\ref{eq:BPhi2}), (\ref{eq: BPhi3}) and (\ref{eq: BPhi4}), we have, for all $\lambda^\frac{1}{4}\leq 0.1$, 
	that 
	\begin{align*}
	\Phi(\lambda)&\ \lesssim\ \lambda^\frac{1}{2}\cdot  \lambda^{\frac{3}{8}} \ =\  \lambda^\frac{7}{8} \qedhere
	\end{align*}
\end{proof}

\begin{theorem}
For the unit ball $\mathbb{B}$, the expected group size after $T$ rounds is $\E[|G^{T}|] \lesssim T^{\frac{1}{8}}$.
\end{theorem}
\begin{proof} 		

Let $t_0$ be a constant such that such that $\frac{\ln(t_0)}{t_0}\leq \frac{1}{10^4}$, then for any round $t\geq t_0$ we have $0\leq \lambda\leq \frac{\ln(t)}{t}\leq \frac{\ln(t_0)}{t_0}\leq \frac{1}{10^4}$. Thus applying Corollary~\ref{cor: MainBound} along with Lemma~\ref{lem: UBBall}  for $t\geq t_0$, we see that
 	\begin{align*}
	\Pr[X^{t+1}]
	 \lesssim  \frac{1}{t}+\int_0^{\frac{\ln(t)}{t}} te^{-t\lambda}  \lambda^{\frac{7}{8}}d\lambda 
	 = \frac{1}{t}+\frac{1}{t^\frac{7}{8}} \int_0^{\ln(t)} e^{-u} \cdot u^{\frac{7}{8}}du
	\lesssim \frac{1}{t^\frac{7}{8}}
	\end{align*}   
Here the equality holds by the substitution $u=t\lambda$.  The last inequality holds since 
$$\int_0^{\ln(t)} e^{-u}u^{\frac{7}{8}}\, du \ \leq\  1+\int_1^{\infty} e^{-u}u \, du  \ \leq\  2$$ 
The theorem follows as $\E[|G^{T}|]=\sum_{t=1}^T  \Pr[X^{t}] \lesssim t_0+ \sum_{t=t_0+1}^T 1/t^{\frac{7}{8}} \lesssim T^\frac{1}{8}$ 
by applying integral bounds.	
\end{proof}

\subsection{Upper Bound for the Unit Square}
Similar to the unit ball case, we must find an appropriate function $f_\mathbb{H}$ satisfying the 
conditions of Corollary~\ref{cor: MainBound}. For a cap $A$ with $A\subseteq H_i(w_1,w_2)$  by Corollary~\ref{cor:subset},
\begin{align}\label{eq: subcap}
\Pr[\ Z(H_i(w_1,w_2)) \ | \ \overbar{H_i(w_1,w_2)}\ ]  \geq  \Pr[\ Z(A)\ |\  \overbar{H_i(w_1,w_2)}\ ]  \geq  \Pr[\ Z(A)\ |\  \overbar{A}\ ]
\end{align}
Let $a(w_1,w_2)\leq b(w_1,w_2)$ be the two side lengths of the triangular cap of greatest area that fits inside 
both $H_i(w_1,w_2)$. Applying Lemma~\ref{lem: f_sq}, along with~(\ref{eq: subcap}) gives
$$\min\limits_i \Pr[\ Z(H_i(w_1,w_2)) \ | \ \overbar{H_i(w_1,w_2)}\ ] \ \gtrsim \ a(w_1,w_2)^4\cdot \ln\left(e\cdot \frac{b(w_1,w_2)}{a(w_1,w_2)}\right)$$
Thus $f_\mathbb{H}(w_1,w_2)=a(w_1,w_2)^4\ln\left(e\frac{b(w_1,w_2)}{a(w_1,w_2)}\right)$ satisfies the conditions of Corollary~\ref{cor: MainBound}.

\begin{lemma}\label{lem: UBSQ}
	Let $a(w_1,w_2)\leq b(w_1,w_2)$  be the two side lengths of the triangular cap of greatest area that fits inside both $H_i(w_1,w_2)$.  
	Then, for any $\lambda\leq \frac{1}{20}$,   
	$$
	\Phi(\lambda) \ =\   \int \int \mathbb{I}\left[a(w_1,w_2)^4\cdot \ln\left(e\cdot \frac{b(w_1,w_2)}{a(w_1,w_2)}\right)\leq \lambda\right] \, dw_1 \, dw_2 
	\ \lesssim\  \lambda\cdot \ln \left(\ln\left( \frac{1}{\lambda}\right)\right) 
	$$	
\end{lemma}
\begin{proof}
	We need only consider pairs $w_1$ and $w_2$ that satisfy the indicator function.
	As $\lambda\leq \frac{1}{20}$, this implies $a(w_1,w_2)\leq \lambda^{\frac14}\le \left(\frac{1}{20}\right)^\frac{1}{4}\leq \frac{1}{2}$ and,
	without loss of generality, $H_1(w_1,w_2)$ is the smallest Voronoi region and fits (under symmetries) into $H_2(w_1,w_2)$.
	Furthermore, applying rotational and diagonal symmetries, any pair of points can be transformed into a pair of the form
	$w_1=(x,y)$ and $w_2=(x+\Delta_x,y+\Delta_y)$, with $s=\Delta_y/\Delta_x\leq 1$ and $\Delta_x,\Delta_y\geq 0$. 
	Hence, we lose only a constant factor in making the following assumptions on $w_1$ and $w_2$: the triangular cap of greatest area that fits inside both the $H_i(w_1,w_2)$ 
	is contained in $H_1(w_1,w_2)$; the cap contains the origin; the larger side corresponding 
	to $b(w_1,w_2)$ is along the $y$-axis; the smaller side corresponding to $a(w_1,w_2)$ is along the $x$-axis. 

	Recall that $H_1(w_1,w_2)$ is either a right-angled triangle or a right-angled trapezoid.
	In the former case, the triangular cap of greatest area which fits inside both of the Voronoi regions is $H_1$ itself. The side lengths of $H_1$ are then the intercepts of $\mathcal{H}$ along the axes, where $\mathcal{H}$ is the hyperplane separating the two Voronoi regions. In the latter case, the triangular cap of greatest area satisfies $b(w_1,w_2)=1$ and $a(w_1,w_2)$ is the intercept of $\mathcal{H}$ on the $x$-axis.
	We can then compute explicit expressions for both terms $a(w_1,w_2)$ and $b(w_1,w_2)$. In particular, 
	\begin{align}
	a(w_1,w_2)=\frac{\|w_2 \|^2-\|w_1 \|^2}{2(w_2-w_1)_1} 
	=x+sy + \frac{\Delta_x}{2} (1+s^2) 
	= x+sy + \frac{\Delta_y}{2}\left(\frac{1+s^2}{s}\right) \label{eq: a-intercept}
	\end{align}
	The first equality holds by definition (\ref{eq:hyperplane}) of hyperplane $\mathcal{H}$. Now, because this is the unit square, we have $b(w_1,w_2)\leq 1$. Thus, $b(w_1,w_2)= \min\left(1,\frac{a(w_1,w_2)}{s}\right)$. Hence, 	
	\begin{align}\label{eq: Amineq}
	a(w_1,w_2)^4 \ln\left(e\cdot\frac{b(w_1,w_2)}{a(w_1,w_2)}\right) 
	\ =\  a(w_1,w_2)^4 \ln\left(e \min\left(\frac{1}{a(w_1,w_2)},\frac{1}{s}\right)\right)
	\end{align}
	For a fixed $w_1=(x,y)$, let $R(x,y)$ be a rectangle containing all the points $w_2=(x+\Delta_x,y+\Delta_y)$ satisfying the condition 
	of the indicator function. Thus, it will suffice to show that we can select $R(x,y)$ to have small area. To do this we must show that $\Delta_x$ and
	$\Delta_y$ cannot be too large.
	Again, recall that if the indicator function is true then $a(w_1,w_2)\leq \lambda^\frac{1}{4}$. So~(\ref{eq: a-intercept}) 
	implies $x\leq \lambda^\frac{1}{4}$ and $s\leq \frac{\lambda^\frac{1}{4}}{y}$. 
	If $\lambda^\frac{1}{4}\leq y\leq1 $ then $\min\left(\frac{1}{a(w_1,w_2)},\frac{1}{s}\right)\geq \frac{y}{\lambda^\frac{1}{4}}$.
	Plugging into~(\ref{eq: Amineq}) gives
	$a(w_1,w_2)^4\cdot \ln\left(e\cdot \frac{b(w_1,w_2)}{a(w_1,w_2)}\right)\geq a(w_1,w_2)^4 \cdot \ln\left(e\cdot \frac{y}{\lambda^\frac{1}{4}}\right) $. It follows that $a(w_1,w_2)\leq \frac{\lambda^{1/4}}{\ln^{1/4}\left(\frac{ey}{\lambda^{1/4}}\right)}$. Therefore, (\ref{eq: a-intercept}) gives: 
	\begin{align}
	\Delta_x&\ \leq\   2\cdot \left(\frac{ \lambda^{\frac{1}{4}}}{\ln^{\frac{1}{4}}\left(\frac{ey}{\lambda^{\frac{1}{4}}}\right)}-ys-x \right)
	\ \ \ \leq\  2\cdot \left(\frac{ \lambda^{\frac{1}{4}}}{\ln^{\frac{1}{4}}\left(\frac{ey}{\lambda^{\frac{1}{4}}}\right)}-x \right)\label{eq:1} \\ 
	\Delta_y&\ \leq\   2s\cdot \left(\frac{ \lambda^{\frac{1}{4}}}{\ln^{\frac{1}{4}}\left(\frac{ey}{\lambda^{\frac{1}{4}}}\right)}-ys-x \right)
	\ \leq\ \frac{1}{2y}\cdot \left(\frac{ \lambda^{\frac{1}{4}}}{\ln^{\frac{1}{4}}\left(\frac{ey}{\lambda^{\frac{1}{4}}}\right)}-x \right)^2 \label{eq:2}
	\end{align}
	The final inequalities in (\ref{eq:1}) and (\ref{eq:2}) were obtained by optimizing over $s\in [0,1]$. Then noting that $x\leq \lambda^\frac{1}{4}$, we have
	\begin{align*}
	\Phi(\lambda)
	&\ \lesssim\ \int_0^1  \int_0^{\lambda^\frac{1}{4}} |R(x,y)|  \, dx \,  dy\\
	&\ =\  \int_0^{\lambda^\frac{1}{4}} \int_0^{\lambda^\frac{1}{4}} |R(x,y)|   \, dx \,  dy + \int_{\lambda^\frac{1}{4}}^1 \int_0^{\lambda^\frac{1}{4}} |R(x,y)|   \, dx \,  dy \\
	&\ \lesssim\ \lambda+ \int_{\lambda^{1/4}}^1 \frac{1}{y}\cdot \int_0^{\frac{ \lambda^{1/4}}{\ln^{1/4}(\frac{ey}{\lambda^{1/4}})}}\, \left(\frac{ \lambda^{1/4}}{\ln^{1/4}(\frac{ey}{\lambda^{1/4}})}-x \right)^3  \, dx  \, dy  \\  
	&\ \lesssim\ \lambda \cdot \int_{\lambda^{1/4}}^1\frac{1}{y} \cdot \frac{1}{\ln(\frac{ey}{\lambda^{1/4}})} \, dy \\
	&\ \lesssim\ \lambda \cdot \ln\left(\ln\left(\frac{1}{\lambda}\right)\right) 
	\end{align*}
	For the second inequality, since $\Delta_x$ must be positive,~(\ref{eq:1}) implies that 
	the limit of the integral becomes $x=\frac{ \lambda^{1/4}}{\ln^{1/4}(\frac{ey}{\lambda^{1/4}})}$. To bound the area $|R(x,y)|$ of 
	the rectangles we have two cases. When $0\leq y\leq \lambda^\frac{1}{4}$, observe, by~(\ref{eq: a-intercept}), that $a(w_1,w_2)\leq \lambda^{\frac{1}{4}}$ implies $\Delta_x\leq 2\lambda^{\frac{1}{4}}$ 
	and $\Delta_y\leq 2\lambda^{\frac{1}{4}}$. Thus $|R(x,y)|\le \Delta_x\cdot \Delta_y\lesssim \sqrt{\lambda}$.
	When $\lambda^\frac{1}{4}\leq y\leq 1$, the bound on $|R(x,y)|$ holds by~(\ref{eq:1}) and~(\ref{eq:2}). 
\end{proof}
	
\begin{theorem}
For the unit square  $\mathbb{H}$, the expected group size after $T$ rounds is $$\E[|G^{T}|] \lesssim \ln T\cdot  \ln \ln T$$
\end{theorem}
\begin{proof}
 Note that for any round $t\geq 100$ we have $0\leq \lambda\leq \frac{\ln(t)}{t}\leq \frac{1}{20}$. Thus, combining Corollary~\ref{cor: MainBound} and Lemma~\ref{lem: UBSQ} we get
	\begin{align*}
	\Pr[X^{t+1}]
	&\ \lesssim \ \frac{1}{t}+\int_0^{\frac{\ln(t)}{t}} te^{-t\lambda} \cdot \lambda \ln \left(\ln\left( \frac{1}{\lambda}\right)\right) \, d\lambda \\ 
	&\ =\ \frac{1}{t}+\frac{1}{t}\int_0^{1} u \ln\left(\ln\left(\frac{t}{u}\right)\right)  du + \frac{1}{t}\int_1^{\ln(t) }u e^{-u}  \ln\left(\ln\left(\frac{t}{u}\right)\right)  du
	\end{align*} 
	Here  the equality holds via the substitution $u=t\lambda$. Note that $u \ln\left(\ln\left(\frac{t}{u}\right)\right)\leq \ln \ln t$, when $t\geq 100$ 
	and $0\leq u\leq 1$. Furthermore,
	$$\int_1^{\ln(t) } ue^{-u}\cdot  \ln\left(\ln\left(\frac{t}{u}\right)\right)\, du \ \leq \ \ln \ln t \cdot \int_1^\infty ue^{-u}\, du \ \leq \ \ln \ln t$$
	Hence it follows $\Pr[X^{t+1}]\lesssim \frac{\ln \ln t}{t} $ for all $t\geq 100$.
	Finally, we see that $\E[|G^{T}|]=\sum_{t=1}^T  \Pr[X^{t}] \lesssim 100+ \sum_{t=101}^T \frac{\ln \ln t}{t} \lesssim \ln T\cdot  \ln \ln T$, 
	where the last inequality was obtained using integral bounds.
\end{proof}

%% file: conclusion.tex
\section{Conclusion}
In this paper we presented techniques for studying the evolution of an exclusive social group in a metric space, under the 
consensus voting mechanism. A natural open problem is to close the gap 
between the $\Omega(\ln T)$ lower bound and the $O(\ln T \cdot \ln\ln T )$ upper bound on the expected cardinality of the 
group, after $T$ rounds, in the unit square. Interesting further directions include the study of higher dimensional metric spaces, 
and allowing for more than two candidates per round. In either direction, our analytic tools may prove useful.

%% file: ConsensusArxivVer.bbl
\begin{thebibliography}{10}

\bibitem{dyn}
N.~Alon, M.~Feldman, Y.~Mansour, S.~Oren, and M.~Tennenholtz.
\newblock Dynamics of evolving social groups.
\newblock {\em ACM Transactions on Economics and Computation}, 7(3):{\#}14,
  2019.

\bibitem{ABJ15}
E.~Anshelevitch, O.~Bhardwaj, and J.~Postl.
\newblock Approximating optimal social choice under metric preferences.
\newblock In {\em Proceedings of the 29th Conference on Artificial Intelligence
  (AAAI)}, pages 777--783, 2015.

\bibitem{AJ17}
E.~Anshelevitch and J.~Postl.
\newblock Randomized social choice functions under metric preferences.
\newblock {\em Journal of Artificial Intelligence Research}, 58(1):797--827,
  2017.

\bibitem{baddeley2007random}
A.~Baddeley, I.~B{\'a}r{\'a}ny, and R.~Schneider.
\newblock Random polytopes, convex bodies, and approximation.
\newblock In Weil. W, editor, {\em Stochastic Geometry}, pages 77--118.
  Springer, 2007.

\bibitem{Bla48}
D.~Black.
\newblock On the rationale of group decision-making.
\newblock {\em Journal of Political Economy}, 56:23--34, 1948.

\bibitem{BLS19}
A.~Borodin, O.~Lev, N.~Shah, and T.~Strangway.
\newblock Primarily about primaries.
\newblock In {\em Proceedings of the 33rd Conference on Artificial Intelligence
  (AAAI)}, pages 1804--1811, 2019.

\bibitem{Cla09}
R.~Claassen.
\newblock Direction versus proximity: Amassing experimental evidence.
\newblock {\em American Politics Research}, 37(2):227--253, 2009.

\bibitem{colomer1999geometry}
J.~Colomer.
\newblock On the geometry of unanimity rule.
\newblock {\em Journal of Theoretical Politics}, 11(4):543--553, 1999.

\bibitem{DHO70}
O.~Davis, M.~Hinich, and P.~Ordeshook.
\newblock An expository development of a mathematical model of the electoral
  process.
\newblock {\em American Political Science Review}, 64:426--448, 1970.

\bibitem{Dow57}
A.~Downs.
\newblock {\em An Economic Theory of Democracy}.
\newblock Harper Collins, 1957.

\bibitem{enelow1984spatial}
J.~Enelow and M.~Hinich.
\newblock {\em The spatial theory of voting: An introduction}.
\newblock Cambridge University Press, 1984.

\bibitem{arrow1990advances}
J.~Enelow and M.~Hinich, editors.
\newblock {\em Advances in the spatial theory of voting}.
\newblock Cambridge University Press, 1990.

\bibitem{FFG16}
M.~Feldman, A.~Fiat, and I.~Golomb.
\newblock On voting and facility location.
\newblock In {\em Proceedings of 17th Conference on Economics and Computation
  (EC)}, pages 269--286, 2016.

\bibitem{Gro73}
B.~Grofman.
\newblock The neglected role of the status quo in models of issue voting.
\newblock {\em The Journal of Politics}, 47:230--237, 1985.

\bibitem{hare1973group}
P.~Hare.
\newblock Group decision by consensus: Reaching unity in the society of
  friends.
\newblock {\em Sociological Inquiry}, 43(1):75--84, 1973.

\bibitem{harold1929stability}
H.~Hotelling.
\newblock Stability in competition.
\newblock {\em Economic Journal}, 39(153):41--57, 1929.

\bibitem{LP10}
D.~Lacy and P.~Paolino.
\newblock Testing proximity versus directional voting using experiments.
\newblock {\em Electoral Studies}, 29(3):460--471, 2010.

\bibitem{LK00}
J.~Lewis and G.~King.
\newblock No evidence on directional vs. proximity voting.
\newblock {\em Politics Analysis}, 8(1):21--33, 2000.

\bibitem{Mat79}
S.~Matthews.
\newblock A simple direction model of electoral competition.
\newblock {\em Public Choice}, 34:141--156, 1979.

\bibitem{merrill1999unified}
S.~Merrill~III, S.~Merrill, and B.~Grofman.
\newblock {\em A unified theory of voting: Directional and proximity spatial
  models}.
\newblock Cambridge University Press, 1999.

\bibitem{Poo05}
K.~Poole.
\newblock {\em Spatial models of parliamentary voting}.
\newblock Cambridge University Press, 2005.

\bibitem{RS89}
G.~Rabinowitz and E.~Stuart.
\newblock A directional theory of issue voting.
\newblock {\em American Political Science Review}, 83:93--121, 1989.

\bibitem{Sch07}
N.~Schofield.
\newblock {\em The spatial models of politics}.
\newblock Routledge, 2007.

\bibitem{SE17}
P.~Skowron and E.~Elkind.
\newblock Social choice under metric preferences: Scoring rules and {STV}.
\newblock In {\em Proceedings of the 31st Conference on Artificial Intelligence
  (AAAI)}, pages 706--712, 2017.

\bibitem{TV08}
M.~Tomz and R.~Van~Houweling.
\newblock Candidate position and voter choice.
\newblock {\em American Political Science Review}, 102(3):303--318, 2008.

\bibitem{vogel1975modern}
E.~Vogel, editor.
\newblock {\em Modern Japanese organization and decision-making}.
\newblock University of California Press, 1975.

\end{thebibliography}
